\documentclass[journal, twoside, web]{ieeecolor}
\usepackage{generic}
\usepackage[english]{babel}

\newif\ifarxiv
\arxivfalse   

\usepackage{xcolor}
\let\labelindent\relax

\usepackage{amsmath,amssymb,amsthm}
\usepackage{enumitem}
\usepackage{mathrsfs}
\usepackage{graphicx}
\usepackage{caption,subcaption}    
\usepackage{setspace}
\usepackage{tipa}
\usepackage{booktabs}
\usepackage{newtxtext}
\usepackage{float}
\usepackage{epsfig}
\usepackage{epstopdf}
\usepackage{array}
\usepackage{algorithm}
\usepackage{algorithmicx}
\usepackage{algpseudocode}
\usepackage{textcomp}
\usepackage{threeparttable}

\usepackage[nocompress]{cite}
\usepackage[colorlinks,linkcolor=blue]{hyperref}
\usepackage[nameinlink,capitalize]{cleveref}
\usepackage{orcidlink}

\theoremstyle{plain}
\newtheorem{proposition}{Proposition}
\newtheorem{theorem}{Theorem}
\newtheorem{lemma}{Lemma}
\newtheorem{corollary}{Corollary}

\theoremstyle{definition}
\newtheorem{definition}{Definition}
\newtheorem{problem}{Problem}
\newtheorem{assumption}{Assumption}
\newtheorem{example}{Example}
\newtheorem{remark}{Remark}

\begin{document}

\title{Submodular Policy Learning for Distributed Task Allocation in Open Multi-Agent Systems}

\author{Jing~Liu, Luca~Ballotta,~\IEEEmembership{Member,~IEEE},
Yangyang~Yang,
Fangfei~Li,~\IEEEmembership{Member,~IEEE},
Yang~Tang,~\IEEEmembership{Fellow,~IEEE},
and~Ruggero~Carli,~\IEEEmembership{Member,~IEEE}
\thanks{This work is supported by the Program of China Scholarship Council under Grant 202506740033, the National Key R\&D Program of China under Grant 2025YFA1016504, the National Natural Science Foundation of China under Grants 62233005, 62573198, U2441245, U25B6002, 62403141, the National Key Laboratory of Space Target Awareness under Grant STA2025ZCB0208, in part by the Shanghai Institute for Mathematics and Interdisciplinary Sciences (SIMIS) under Grant SIMIS-ID-2025-SP. \textit{(Corresponding authors: Fangfei Li and Yang Tang.)}}
\thanks{Jing Liu and Yangyang Yang are with the School of Mathematics, East China University of Science and Technology, Shanghai 200237, China (e-mail: y20220091@mail.ecust.edu.cn; y20250091@mail.ecust.edu.cn).}
\thanks{Luca Ballotta and Ruggero Carli are with the Department of Information Engineering, University of Padova, 35131 Padova, Italy (e-mail: luca.ballotta@unipd.it; ruggero.carli@unipd.it).}
\thanks{Fangfei Li is with the School of Mathematics and the Key Laboratory of Smart Manufacturing in Energy Chemical Process, Ministry of Education, East China University of Science and Technology, Shanghai 200237, China (e-mail: li\_fangfei@163.com, lifangfei@ecust.edu.cn).}
\thanks{Yang Tang is with the Key Laboratory of Smart Manufacturing in Energy Chemical Process, Ministry of Education, East China University of Science and Technology, Shanghai 200237, China (e-mail: yangtang@ecust.edu.cn).}}

\maketitle

\begin{abstract}
This paper studies policy learning for distributed task allocation in open multi-agent systems, where agents may join and leave in a time-varying fashion, with submodular stage team  utilities. At each time, the active agents select actions from local categorical policies such that the feasible joint agent-action pairs form a partition matroid. Standard continuous relaxations of submodular set functions are based on independent Bernoulli sampling, making them inconsistent with agents' policies.
To solve this mismatch, we propose the \emph{partition multilinear extension} (PME), a policy-based relaxation whose  continuous support matches feasible actions under categorical policies.
We prove that the marginal gains of the stage utility provide an unbiased estimator of the gradient of the PME and that maximizing the PME over action distributions is equivalent to maximizing the stage utilities over agent actions, which are critical to devise principled policy gradient.
Building on this, we design \emph{SubMAPL}, a centralized-training decentralized-execution KL-mirror policy-learning method that uses local marginal gains as stochastic PME gradients during training. KL-mirror updates preserve categorical feasibility without Euclidean projection.
In the case where agents run tabular-softmax policies, we introduce open policy migration and an open-system KL tracking variation to handle agent arrivals and departures. Using dynamic regret analysis, we establish a lower bound on the cumulative utility which accounts for the openness of the environment and for the gap between optimal stage-wise and global utilities. Simulations on multi-agent coverage demonstrate that SubMAPL outperforms policy-gradient and online-learning baselines.

\end{abstract}


\section{Introduction}
\label{sec:introduction}

\IEEEPARstart{M}{ulti}-agent task allocation in open multi-agent systems (OMAS) arises in dynamic target tracking, active sensing, and coverage, where agents must coordinate local decisions under time-varying participation, limited communication, and non-additive team utilities. In long-duration deployments, agents may enter, leave, fail, or recover due to hardware malfunctions, battery recharging, or mission-level reconfiguration ~\cite{zhang2018fully,xu2023online,deplano2026optimization}. 

A central feature of such tasks is that utility functions are rarely additive across agents. For instance, overlapping field-of-views capture redundant information in coverage tasks. This so-called diminishing-returns structure is naturally modeled by monotone submodular set functions~\cite{fujishige2005submodular,Welikala2022ANP,xu2025communication}. Submodularity has been extensively used in combinatorial optimization because it enables approximation guarantees for greedy selection under matroid constraints~\cite{nemhauser1978analysis,fisher1978analysis, xu2023online,liu2026distributed} and the use of continuous relaxations such as the multilinear extension (MLE) with rounding~\cite{calinescu2011maximizing, zhang2025near}. While these approaches are mostly tailored to static centralized optimization and task allocation and require dense communication graphs for consensus, the benefits of submodularity in learning data-driven multi-agent policies are underexplored.

Multi-agent reinforcement learning (MARL) is widespread to learn decentralized agent policies under partial observability and dynamical  environment~\cite{zhang2018fully,lowe2017multi,yu2022surprising}. Although the interplay between MARL and submodularity has been studied, existing results are limited to centralized settings~\cite{prajapat2024submodular,
desanti2024global,chen2026multi}.
On the other hand, OMASs have garnered attention in the online learning literature, where submodularity-based methods enable suboptimality quantification of online optimization algorithms rather than learnable policies~\cite{xu2023online,zhang2025near,zhang2025effective}. 

Incorporating submodular utilities into the learning of distributed policies is challenging. 
Static submodular problems are NP-hard in general and admit approximation guarantees rather than exact polynomial-time solutions~\cite{nemhauser1978analysis,fisher1978analysis,calinescu2011maximizing}.
The data-driven and dynamic nature of multi-agent policy learning introduces additional challenges compared to static optimization.
First, the gradient feedback signal used in training must be consistent with decentralized categorical policies. However,
classical multilinear relaxations rely on independent Bernoulli sampling and rounding, and therefore do not represent factorized categorical policies executed by agents~\cite{calinescu2011maximizing,zhang2025near}.
Moreover, the combinatorial difficulty is coupled with sequential state evolution, decentralized information, and exponentially growing joint action space which demands careful credit assignment~\cite{zhang2018fully,lowe2017multi,yu2022surprising}.
Existing submodular RL formulations target centralized policies rather than fully distributed decision-making~\cite{prajapat2024submodular,desanti2024global,chen2026multi}. 

\subsubsection*{Contribution}
Building on the PME and its gradient characterization introduced in our preliminary work~\cite{liu2026submodular},
we propose three main contributions to bridge the above mentioned gaps.
\begin{enumerate}[label=(\arabic*), leftmargin=*]

\item We formulate distributed online submodular coordination as Markov decision process (MDP) in an OMAS and develop \emph{SubMAPL}, a KL-mirror policy-learning method for decentralized submodular task allocation. 
Unlike sampled-action difference-reward policy gradients and counterfactual credit-assignment methods~\cite{foerster2018coma,castellini2025difference,liu2026submodular}, SubMAPL uses a stochastic local gradient that evaluates every feasible local action of each active agent under the same sampled actions of the other agents. We show that this vector is an unbiased estimator of the PME coordinate gradient, thereby providing a principled first-order interpretation of marginal-contribution feedback. 

\item We introduce an open policy-migration mechanism to handle agent arrivals and departures, and define an open-system KL tracking variation that measures changes in both the stagewise optimal comparator and the time-varying policy domain. Our preliminary work~\cite{liu2026submodular} establishes guarantees for idealized Euclidean projected dynamics that directly update the PME marginals. However, the standard softmax policy-gradient implementation updates the logits, and the resulting marginal change depends on the softmax Jacobian and agrees with PME ascent only to first order.
SubMAPL admits an exact equivalence between its KL-mirror policy update and an additive tabular-logit update. This removes the first-order mismatch between the analyzed policy-space dynamics and the implemented parameter update, allowing us to establish open-system guarantees directly for the policy sequence generated under agent arrivals and departures.

\item  By leveraging dynamic regret analysis in conjunction with DR-submodularity and monotonicity of the PME, we establish a lower bound on the cumulative utility within a $1/2$-approximation factor from the optimum.
We further establish a $1/3$-factor bound in the case where the stage utilities are themselves marginal gains of a utility function capturing performance through the whole horizon, including scenarios such as coverage of static maps. These results complement previous analysis on stability and adaptability of open systems~\cite{deplano2026optimization, anil2025mohito}.
Our analysis significantly improves our preliminary work~\cite{liu2026submodular}, which derives a stagewise bound on the PME and a regret bound on the cumulative PME without formal guarantees on the original finite-horizon cumulative utility.
\end{enumerate}

\subsubsection*{Organization}
~\cref{sec:problem} formulates distributed submodular
coordination as an MDP in OMAS.
~\cref{sec:pme} introduces the PME and
its properties. ~\cref{sec:algorithm} presents SubMAPL and ~\cref{sec:theory} derives its suboptimality bound.
~\cref{sec:experiments} reports simulations, and ~\cref{sec:conclusion} concludes the paper.

\subsubsection*{Notation}
We use $\mathbb N$ to denote the set of positive integers and define $[T]=\{1,\ldots,T\}$ for any $T\in\mathbb N$. For a finite set $S$, $|S|$ denotes its cardinality and $2^S$ denotes its power set. For sets $A$ and $B$, $A\setminus B$, $A\cup B$, and $A\cap B$ denote set difference, union, and intersection, respectively. For a scalar $z$, $[z]_+=\max\{z,0\}$. For vectors $x$ and $y$, $x\le y$ denotes componentwise inequality and $\langle x,y\rangle$ denotes the Euclidean inner product. We use $\nabla f(x)$ and $\partial f/\partial x_i$ for gradients and partial derivatives,
respectively. Expectations and probabilities are denoted by
$\mathbb E[\cdot]$ and $\Pr(\cdot)$. The Kullback--Leibler divergence is denoted by $D_{\rm KL}(\cdot\Vert\cdot)$. Standard asymptotic notation is denoted by $O(\cdot)$ and $o(\cdot)$. \cref{tab:notation} summarizes the main paper-specific notation.

\begin{table}[t]
\centering
\caption{Main notation in the paper.}
\label{tab:notation}
\small
\begin{tabular}{@{}p{0.25\columnwidth}p{0.7\columnwidth}@{}}
\toprule
Symbol & Meaning \\
\midrule
$\mathcal M$ & Markov decision process in OMAS. \\
$\mathcal N_t$ & Set of active agents at time $t$. \\
$\mathcal R_t$ & Set of agents active at both times $t-1$ and $t$. \\
$\mathcal N_t^{\rm in}$ & Set of agents arriving at time $t$. \\
$\mathcal N_t^{\rm out}$ & Set of agents departing at time $t$. \\
$\mathcal S$ & Global state space. \\
$s_t$ & Global state at time $t$. \\
$\mathcal A_{i}$ & Local action space of agent $i$. \\
$o_{i,t}$ & Local observation of agent $i$ at time $t$. \\
$a_{i,t}$ & Local action selected by agent $i$ at time $t$. \\
$\mathbf{a}_t$ & Joint action tuple $(a_{i,t})_{i\in\mathcal N_t}$. \\
$\Omega_t$ & Active agent-action space at time $t$. \\
$\Omega_{i,t}$ & Local agent-action space of agent $i$ at time $t$. \\
$\mathcal I_t$ & Partition-matroid feasible family on $\Omega_t$. \\
$A_t$ & Set of agent-action pairs executed at time $t$. \\
$F_t(A_t;s_t)$ & Stage utility at time $t$. \\
$\pi^i$ & Local categorical policy of agent $i$. \\
$\pi$ & Sequence of factorized policies with active agents.\\
$p^\pi$ & Distribution of state-action trajectory under $\pi$. \\
$J(\pi)$ & Expected cumulative utility under  $\pi$. \\
$x_t^\pi(s_t)$ & Policy-induced action marginal vector at state $s_t$. \\
$\mathcal P(\mathcal I_t)$ & Partition-matroid polytope at time $t$. \\
$\mathcal F_t$ & Categorical-policy face of $\mathcal P(\mathcal I_t)$. \\
$\tilde f_t(\cdot;s_t)$ & Partition multilinear extension of $F_t(\cdot;s_t)$. \\
$\mathcal D_t(\cdot\mid x)$ & Partition sampling distribution induced by $x$. \\
$\mathcal D_{t}^{-i}(\cdot\mid x)$ & Sampling distribution over all agents except $i$. \\
$\widehat g_t$ & Stochastic estimator of $\nabla\widetilde f_t(x_t;s_t)$. \\
$\mathcal H_t$ & Pre-sampling history at time $t$. \\
$x_t^\star$ & PME maximizer at time $t$. \\
$\mathcal V_T^{\rm open}$ & Open-system KL tracking variation. \\
$\mathrm{Reg}_T^{1/2}$ & Dynamic $1/2$-approximation regret.\\
\bottomrule
\end{tabular}
\end{table}

\section{System Setup and Problem Formulation}
\label{sec:problem}

\begin{figure}
    \centering
    \includegraphics[width=1\linewidth]{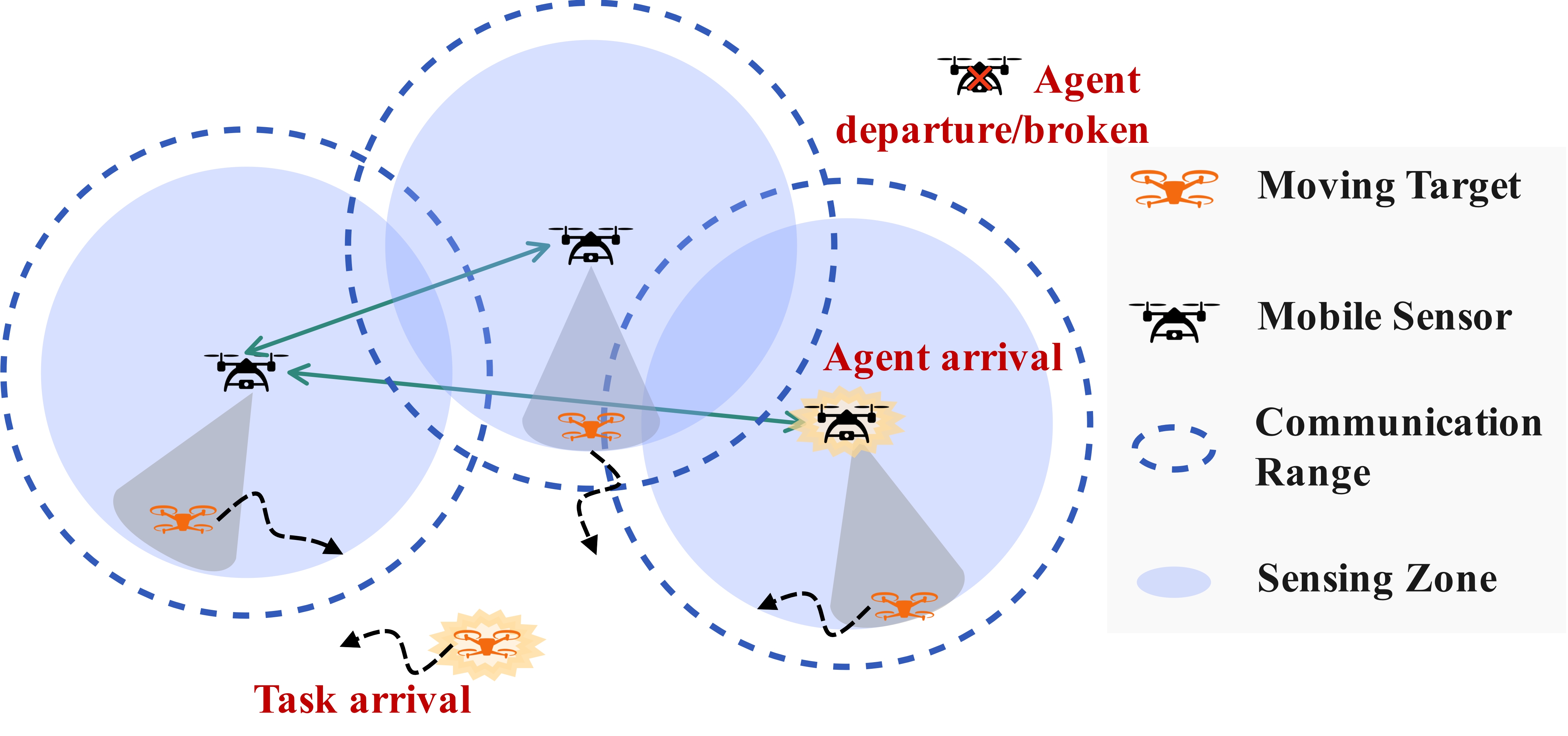}
    \captionsetup{font={small}}
      \caption{Open multi-agent task allocation with time-varying agent participation, dynamic tasks, and limited communication and sensing.}
    \label{fig:system}
\end{figure}

We study decentralized task allocation in open MDPs~\cite{prajapat2024submodular,deplano2026optimization}. At each time step, the active agents locally observe the state of the environment and execute actions aiming to maximize a non-additive submodular global utility. Active agents and utilities may change over time, reflecting agent arrivals and departures, dynamic tasks, or an evolving environment \cite{LingfeiSu2026OptimalTO}, as shown in \autoref{fig:system}.
Our goal is to efficiently learn effective decentralized policies by leveraging the submodularity of team utilities.

\subsection{Markov Decision Process in OMAS}

Let $t\in[T]=\{1,\ldots,T\}$ denote the time step and
$\mathcal N_t$ the set of active agents at time $t$.
We model the task-allocation process as a finite-horizon partially observed open multi-agent Markov decision process $\mathcal M = \langle \mathcal S,\mathcal A,\mathcal O,\mathcal P,F,\mu_1 \rangle$, where $\mathcal S$ is the global state space, $\mathcal A$ denotes the collection of feasible local action sets, $\mathcal O=\{O_t\}_{t\in[T]}$ is the sequence of local observation maps, $\mathcal P=\{P_t\}_{t\in[T]}$ is the sequence of controlled transition kernels, $F=\{F_t\}_{t\in[T]}$ is the sequence of stage-utility set functions, and $\mu_1$ is the initial-state distribution on $\mathcal S$.
For each agent $i$ that may be active during the horizon, let $\mathcal O_i$ denote its finite local observation space. Here, $O_t=(O_{i,t})_{i\in\mathcal N_t}$, where $O_{i,t}:\mathcal S\to\mathcal O_i$ and
$o_{i,t}=O_{i,t}(s_t)$. The initial state is sampled as $s_1\sim\mu_1$~\cite{puterman2014markov,deplano2026optimization,chen2026multi}.

Each active agent $i\in\mathcal{N}_t$ receives a local observation $o_{i,t}$ and selects an action $a_{i,t}$ from a finite nonempty feasible action set $\mathcal A_{i}$.
The global state $s_t\in\mathcal S$ characterizes the environment, including the tasks to be completed. Given the joint action $\mathbf{a}_t=(a_{i,t})_{i\in\mathcal N_t}$, the transition kernel $P_t$ governs the state evolution as $s_{t+1}\sim P_t(\cdot\mid s_t,\mathbf{a}_t)$.

The ground set of agent-action pairs available at time $t$ is $\Omega_t = \bigl\{(i,a): i\in\mathcal{N}_t,\ a\in\mathcal A_{i}\bigr\}$.
For each active agent, define the local actions $\Omega_{i,t} =
\bigl\{(i,a): a\in\mathcal A_{i}\bigr\},\ i\in\mathcal{N}_t$.
The feasible subsets of agent-action pairs are described by the partition matroid $(\Omega_t,\mathcal I_t)$ with
$\mathcal I_t =  \bigl\{  A\subseteq\Omega_t:
 |A\cap\Omega_{i,t}|\le 1,\ \forall i\in\mathcal N_t \bigr\}$. The executed joint action induced by the local choices is $A_t= \{(i,a_{i,t}):i\in\mathcal N_t\}\in\mathcal I_t$ \cite{calinescu2011maximizing}.

To describe changes in the open-system domain, let $\mathcal{N}_0=\emptyset$ and define
 $\mathcal R_t
  =  \mathcal N_t\cap\mathcal N_{t-1}, 
  \mathcal N_t^{\rm in}
  =  \mathcal N_t\setminus\mathcal N_{t-1},
  \mathcal N_t^{\rm out}
  =  \mathcal N_{t-1}\setminus\mathcal N_t$. 
Here $\mathcal{R}_t$, $\mathcal{N}_t^{\rm in}$, and $\mathcal{N}_t^{\rm out}$ denote the remaining, arriving, and departing agents at time $t$, respectively. Moreover, we assume the agent arrivals and departures are independent of the agents' local policies and the state transitions~\cite{deplano2026optimization}. 

\subsection{Decentralized Categorical Policies}

Each active agent $i\in\mathcal N_t$ processes a local observation $o_{i,t}=O_{i,t}(s_t)$, where $O_{i,t}:\mathcal S\to\mathcal O_i$ is a deterministic observation map modeling local sensing and communication with neighbors, and selects an action $a_{i,t}\in\mathcal A_i$. A policy $\pi^i$ assigns a categorical distribution over $\mathcal A_{i}$ given $o\in\mathcal O_i$ such that, at time $t$, active agent $i\in\mathcal N_t$ computes its action as a realization $a_{i,t}\sim\pi^i(\cdot\mid o_{i,t})$.
Let $\mathbf o_t=(o_{i,t})_{i\in\mathcal N_t}$ denote the joint
observation vector.
We consider factorized categorical policies \cite{oliehoek2016concise,zhang2018fully}
\begin{equation}
\label{eq:factorized_policy}
\pi_t(\mathbf a_t\mid\mathbf o_t)
=
\prod_{i\in\mathcal N_t}
\pi^i(a_{i,t}\mid o_{i,t}).
\end{equation}
The factorized policy $\pi_t$ is time dependent as a consequence of time-varying active agent sets $\mathcal{N}_t$.
The factorization in~\eqref{eq:factorized_policy} governs decentralized execution during deployment;
since the distributions
$\{\pi^i(\cdot\mid o_{i,t})\}_{i\in\mathcal N_t}$ are conditioned on local observations $o_{i,t}$, the agents sample their actions independently.

\subsection{Submodular Team Utility and Problem Statement}

At each time step $t$, the team performance is measured by a set function $F_t(\cdot\,;s_t):2^{\Omega_t}\to\mathbb R_{\ge0}$ defined over subsets of the current agent-action ground set $\Omega_t$. 

\begin{definition}[Normalized monotone submodular set function~\cite{krause2014SubmodularFM}]
\label{def:submodular_set_function}
Let $\Omega$ be a finite ground set. For any $A\subseteq\Omega$ and $e\in\Omega\setminus A$, define the marginal gain of adding $e$ to $A$ as $F(e\mid A) = F(A\cup\{e\})-F(A)$. A set function $F:2^\Omega\to\mathbb R_{\ge 0}$ is normalized, monotone, and submodular if it respectively satisfies the following properties:
\begin{enumerate}
    \item $F(\emptyset)=0$.
    \item for any sets $A\subseteq A'\subseteq\Omega$, $F(A)\le F(A')$.
    \item for any sets $A\subseteq A'\subseteq\Omega$ and element $e\in\Omega\setminus A'$, $F(e\mid A)\ge F(e\mid A')$.
\end{enumerate}
\end{definition}

\begin{assumption}[Utility and Feasibility Conditions]
\label{ass:standing_conditions}
For every time step $t\in[T]$, the following holds.
\begin{enumerate}
    \item[(i)] For every state $s_t$, the team utility
    $F_t(\cdot\,;s_t):2^{\Omega_t}\to\mathbb R_{\ge0}$
    is normalized, monotone, and submodular.
    
    \item[(ii)] Each active agent has a nonempty local feasible action set, i.e.,
    $\mathcal A_{i}\neq\emptyset$ for all $i\in\mathcal N_t$.
\end{enumerate}
\end{assumption}

\cref{ass:standing_conditions} means that adding an agent-action pair can increase the team utility but its marginal gain decreases when more agent-action pairs have already been selected. This diminishing-returns structure captures redundancy among agent contributions~\cite{xu2023online}.

\begin{problem}[Open-System Submodular Policy Learning]
\label{prob:main}
Given the open multi-agent MDP, find a decentralized categorical policy $\pi = \{\pi_t\}_{t=1}^T$, with $\pi_t=\prod_{i\in\mathcal N_t}\pi^i$, that maximizes the expected cumulative submodular utility

\begin{equation}
\label{eq:objective}
  \begin{aligned}
   \max_{\pi=\{\pi_t\}_{t=1}^T}\quad & J(\pi)
      =
      \mathbb E_{\tau\sim p^\pi}
      \left[
      \sum_{t=1}^{T} F_t(A_t;s_t)
      \right] \\
    \mathrm{s.t.}\quad
    & a_{i,t}\sim \pi^i(\cdot\mid o_{i,t}), 
      \ \forall i\in\mathcal N_t, \ \forall t\in[T],\\
    & A_t\in\mathcal I_t,
      \ \forall t\in[T].
  \end{aligned}
\end{equation}
Here, $\tau=(s_1,\mathbf a_1,s_2,\ldots,
s_T,\mathbf a_T,s_{T+1})$ and $p^\pi$ denotes the trajectory distribution induced by $\pi$ and the transition kernels $\{P_t\}_{t=1}^{T}$.
\end{problem}

Even for a single time step, \cref{prob:main} reduces to monotone submodular maximization under a partition matroid constraint, which is NP-hard in general~\cite{nemhauser1978analysis,calinescu2011maximizing}. 
Over a finite horizon, this combinatorial difficulty is exacerbated by action-dependent state evolution and time-varying agent participation. 
Since directly optimizing the cumulative utility (i.e., the return) is intractable~\cite{prajapat2024submodular}, we adopt a stagewise approach and leverage the submodularity of stage utilities to learn the decentralized policies. 
The next section develops a continuous relaxation of stage utilities that formally supports our stagewise policy learning framework.

\section{Partition Multilinear Extension}
\label{sec:pme}

A standard approach to combinatorial submodular  maximization is to replace the discrete set problem with a continuous relaxation that is amenable to continuous optimization methods, particularly under matroid constraints~\cite{vondrak2008optimal,calinescu2011maximizing}. The standard multilinear extension is based on independent Bernoulli sampling over ground-set elements~\cite{chekuri2014submodular}. Under the partition constraint induced by $\mathcal I_t$, this sampling rule may select multiple elements from the same agent action subset $\Omega_{i,t}$, creating a formal mismatch with the policy in \eqref{eq:factorized_policy} which selects one feasible action per agent. This motivates the continuous relaxation defined in the following to support training of factorized categorical policies in a principled fashion.

Building on our preliminary work~\cite{liu2026submodular}, we use the
partition multilinear extension (PME), a continuous relaxation of utilities $F_t$ whose sampling distribution matches the factorized categorical execution under the partition matroid.
At each time step $t$, policy~\eqref{eq:factorized_policy} at observations $\{o_{i,t}\}_{i\in\mathcal{N}_t}$ induces a marginal vector $x_t^\pi(s_t)\in[0,1]^{|\Omega_t|}$ over agent-action pairs, with coordinates
\begin{equation}
\label{eq:policy_induced_marginal}
  x_{(i,a),t}^{\pi}(s_t)  = \pi^i(a\mid o_{i,t}), \
  (i,a)\in\Omega_t.
\end{equation}
For readability, we denote a generic marginal vector by $x$. Each coordinate $x_{(i,a),t}^\pi(s_t)$ is the probability that agent $i\in\mathcal{N}_t$ executes action $a$ at the current state $s_t$ under policy $\pi^i$.

\subsection{Definition and Policy Equivalence}
\label{sec:pme_definition}

The partition matroid polytope associated with $(\Omega_t,\mathcal I_t)$ is \cite{calinescu2011maximizing}
\begin{equation}
\label{eq:partition_matroid_polytope}
  \mathcal P(\mathcal I_t)
  = 
  \left\{
  x\in[0,1]^{|\Omega_t|}:
  \sum_{a\in\mathcal A_{i}} x_{(i,a)}\le 1,\ 
  \forall i\in\mathcal N_t
  \right\}.
\end{equation}
For $x\in\mathcal P(\mathcal I_t)$, let
$\mathcal D_t(\cdot\mid x)$ denote the partition sampling
distribution over $\mathcal I_t$. Under this distribution, each agent $i\in\mathcal N_t$ selects element $(i,a)$ with probability $x_{(i,a)}$, and selects no element with probability $1-\sum_{a\in\mathcal A_{i}}x_{(i,a)}$. The choices are independent
across agents.

\begin{definition}[Partition Multilinear Extension \cite{liu2026submodular}]
\label{def:pme}
Given the current state $s_t$ and open-system domain, the partition
multilinear extension of $F_t(\cdot;s_t)$ is the function
$\tilde f_t(\cdot;s_t):\mathcal P(\mathcal I_t)\to\mathbb R_{\ge0}$
defined by
\begin{align}
\tilde f_t(x;s_t)
&=
\mathbb E_{A\sim\mathcal D_t(\cdot\mid x)}
\left[
F_t(A;s_t)
\right]
\label{eq:pme_def_expectation}
\\
&=
\sum_{A\in\mathcal I_t}
F_t(A;s_t)
\prod_{i\in\mathcal N_t}
p_i(A;x),
\label{eq:pme_def_sum}
\end{align}
where
\begin{equation}
\label{eq:agent_prob}
  p_{i}(A;x)
  =
  \begin{cases}
  x_{(i,a)},
  & \text{if } A\cap\Omega_{i,t}=\{(i,a)\}, \\[3pt]
  1-\displaystyle\sum_{a'\in\mathcal A_{i}}x_{(i,a')},
  & \text{if } A\cap\Omega_{i,t}=\emptyset.
  \end{cases}
\end{equation}
\end{definition}

The categorical-policy face of $\mathcal P(\mathcal I_t)$ is
\begin{equation}
  \mathcal F_t
  =
  \left\{
  x\in\mathcal P(\mathcal I_t):
  \sum_{a\in\mathcal A_{i}}x_{(i,a)}=1,\ 
  \forall i\in\mathcal N_t
  \right\}.
\end{equation}
For each $i\in\mathcal N_t$, define the local categorical simplex $\mathcal F_{i} =  \left\{
  x_i\in[0,1]^{|\mathcal A_{i}|}:
  \sum_{a\in\mathcal A_{i}}x_{(i,a)}=1
  \right\}$.
Then $\mathcal F_t=\prod_{i\in\mathcal N_t}\mathcal F_{i}$. For any $x\in\mathcal F_t$, the probability of selecting no action is zero in \eqref{eq:agent_prob}. Hence, the distribution induced by the PME respects the partition structure of $\Omega_t$ since it selects exactly one agent-action pair from each $\Omega_{i,t}$ w.p.1, matching the categorical action-selection model in \eqref{eq:factorized_policy}.

\begin{remark}
At a fixed time step, the PME is closely related to the
policy-based continuous extension in \cite{zhang2025effective} and the Multinoulli Extension in \cite{zhang2025multinoulli}. The former considers online coordination over fixed agent and action sets, whereas the latter addresses static subset selection under fixed partition constraints. Since these formulations do not model controlled Markovian dynamics or changes in the active-agent set, they cannot capture the coupling among current decisions, future states, and utilities. Our formulation extends the PME to finite-horizon decentralized policy learning for an OMAS and establishes exact equivalence with both the expected stage utility and the finite-horizon policy objective.
\end{remark}

\begin{lemma}[Policy-PME Objective Equivalence~\cite{liu2026submodular}]
\label{lem:policy_pme_equivalence}
Let $\pi_t=\prod_{i\in\mathcal N_t}\pi^i$ be a factorized categorical policy, and let $x_t^\pi(s_t)$ be defined by \eqref{eq:policy_induced_marginal}. Then, for any fixed $s_t$,
\begin{equation}
\label{eq:policy_pme_equivalence}
\mathbb E_{A_t\sim\pi_t}
\left[
F_t(A_t;s_t)\mid s_t
\right]
=
\tilde f_t\bigl(x_t^\pi(s_t);s_t\bigr).
\end{equation}
\end{lemma}

\begin{proof}
Consider the current state $s_t$ and the current open-system domain. Then the local observations $\{o_{i,t}\}_{i\in\mathcal N_t}$ are fixed, and each agent samples independently from
$\pi^i(\cdot\mid o_{i,t})$. Hence, the executed set is $ A_t=\{(i,a_{i,t}):i\in\mathcal N_t\}, \ a_{i,t}\sim\pi^i(\cdot\mid o_{i,t})$. By definition of $x_t^\pi(s_t)$, the per-agent factor in \eqref{eq:agent_prob} satisfies $ p_i(A;x_t^\pi(s_t)) = \pi^i(a\mid o_{i,t}) \ \text{if } A\cap\Omega_{i,t}=\{(i,a)\}$. Moreover, since $\pi^i(\cdot\mid o_{i,t})$ is a categorical distribution, $ \sum_{a\in\mathcal A_{i}}x^\pi_{(i,a),t}(s_t) =  1$, the probability of selecting no action in \eqref{eq:agent_prob} is zero for every active agent. Therefore, under $\mathcal D_t(\cdot\mid x_t^\pi(s_t))$, positive probability
is assigned only to feasible sets that select exactly one action per active agent, and this probability equals the product of the corresponding local categorical probabilities. Thus, for any feasible $A\in\mathcal I_t$,  $\Pr(A_t=A\mid s_t)  =  \prod_{i\in\mathcal N_t}   p_i(A;x_t^\pi(s_t))$. Substituting this identity into the conditional expectation gives
\begin{align*}
  \mathbb E_{A_t\sim\pi_t}
  \left[
  F_t(A_t;s_t)\mid s_t
  \right]
  &=
  \sum_{A\in\mathcal I_t}
  F_t(A;s_t)
  \Pr(A_t=A\mid s_t) \\
  &=
  \sum_{A\in\mathcal I_t}
  F_t(A;s_t)
  \prod_{i\in\mathcal N_t}
  p_i(A;x_t^\pi(s_t)) \\
  &=
  \tilde f_t\bigl(x_t^\pi(s_t);s_t\bigr),
\end{align*}
which proves \eqref{eq:policy_pme_equivalence}. 
\end{proof}

Let $\mathcal D_{t}^{-i}(\cdot\mid x)$ denote the joint distribution induced by $x$ over the local action blocks of all agents except $i$, namely $\{\Omega_{j,t}\}_{j\in\mathcal N_t\setminus\{i\}}$.
\begin{lemma}[PME Coordinate Gradient~\cite{liu2026submodular}]
\label{lem:pme_coordinate_gradient}
For any $(i,a)\in\Omega_t$ and $x\in\mathcal P(\mathcal I_t)$,
\begin{equation}
\label{eq:pme_coordinate_gradient}
\frac{\partial \tilde f_t}{\partial x_{(i,a)}}(x;s_t)=
\mathbb E_{A^{-i}\sim\mathcal D_{t}^{-i}(\cdot\mid x)}
\left[
F_t\bigl((i,a)\mid A^{-i};s_t\bigr)
\right],
\end{equation}
where
$F_t\bigl((i,a)\mid A^{-i};s_t\bigr)=
F_t\bigl(A^{-i}\cup\{(i,a)\};s_t\bigr)-F_t(A^{-i};s_t)$
denotes the marginal gain of adding the agent-action pair $(i,a)$ to $A^{-i}$ at state $s_t$.
\end{lemma}

\begin{proof}
Consider an arbitrary $(i,a)\in\Omega_t$. In the sum representation in \eqref{eq:pme_def_sum}, the coordinate $x_{(i,a)}$ appears only in the local factor $p_i(A;x)$. There are three cases. If $A\cap\Omega_{i,t}=\{(i,a)\}$, then $\partial p_i(A;x)/\partial x_{(i,a)}=1$. If $A\cap\Omega_{i,t}=\{(i,a')\}$ for some $a'\ne a$, then $\partial p_i(A;x)/\partial x_{(i,a)}=0$. If $A\cap\Omega_{i,t}=\emptyset$, then $\partial p_i(A;x)/\partial x_{(i,a)}=-1$. Applying the product rule to \eqref{eq:pme_def_sum} gives
\begin{align*}
  \frac{\partial \tilde f_t}{\partial x_{(i,a)}}(x;s_t)
  &=
  \sum_{A:(i,a)\in A}
  F_t(A;s_t)
  \prod_{k\ne i}p_k(A;x)  \\
  &\quad -
  \sum_{A:A\cap\Omega_{i,t}=\emptyset}
  F_t(A;s_t)
  \prod_{k\ne i}p_k(A;x).
\end{align*}
Every set in the first sum can be written uniquely as $A^{-i}\cup\{(i,a)\}$, while every set in the second sum can be identified with the same partial joint action $A^{-i}$ of all agents except $i$. Moreover, $\prod_{k\ne i}p_k(A;x)$ is exactly the probability of $A^{-i}$ under $\mathcal D_{t}^{-i}(\cdot\mid x)$. Hence
\begin{align*}
  & \frac{\partial \tilde f_t}{\partial x_{(i,a)}}(x;s_t)\\
  &=
  \sum_{A^{-i}}
  \Pr_{\mathcal D_{t}^{-i}(\cdot\mid x)}(A^{-i})
  \left[
  F_t(A^{-i}\cup\{(i,a)\};s_t)
  -
  F_t(A^{-i};s_t)
  \right]  \\
  &=
  \mathbb E_{A^{-i}\sim\mathcal D_{t}^{-i}(\cdot\mid x)}
  \left[
  F_t\big((i,a)\mid A^{-i};s_t\big)
  \right].
\end{align*}
This proves \eqref{eq:pme_coordinate_gradient}.
\end{proof}

\cref{lem:pme_coordinate_gradient} shows that partial derivative of PME is the expected marginal gain of the corresponding agent-action pair under the current categorical sampling distribution of the other agents. This identity provides the first-order information used by the counterfactual marginal-gain estimator in \cref{sec:algorithm}.

\begin{proposition}[Properties of the PME~\cite{liu2026submodular}]
\label{pro:pme_properties}
Under \cref{ass:standing_conditions},
$\tilde f_t(\cdot;s_t)$ satisfies the following properties on
$\mathcal P(\mathcal I_t)$:
\begin{enumerate}
  \item $\tilde f_t(0;s_t)=0$.
  \item $\tilde f_t(x;s_t)\ge0$ for all $x\in\mathcal P(\mathcal I_t)$.
  \item $\tilde f_t(\cdot\,;s_t)$ is a multilinear polynomial and hence is smooth.
  \item $\nabla \tilde f_t(x;s_t)\ge0$ componentwise for all $x\in\mathcal P(\mathcal I_t)$.
  \item For any $x,y\in\mathcal P(\mathcal I_t)$ with $x\le y$ componentwise,
  $\nabla \tilde f_t(x;s_t)\ge\nabla \tilde f_t(y;s_t)$ componentwise.
\end{enumerate}
Consequently, $\tilde f_t(\cdot\,;s_t)$ is monotone and DR-submodular on $\mathcal P(\mathcal I_t)$.
\end{proposition}

\begin{proof}
We prove the properties in order.

1) 
When $x=0$, each agent selects no element with probability one under $\mathcal D_t(\cdot\mid 0)$. Therefore, the only set with nonzero probability is $\emptyset$, and $\tilde f_t(0;s_t)=F_t(\emptyset;s_t)=0$.

2) By nonnegativity of $F_t$ and by the fact that all factors $p_i(A;x)$ are valid probabilities for $x\in\mathcal P(\mathcal I_t)$, every term in \eqref{eq:pme_def_sum} is nonnegative. Thus, $\tilde f_t(x;s_t)\ge0$.

3) For each fixed set $A\in\mathcal I_t$, every factor $p_i(A;x)$ is affine in the coordinates $\{x_{(i,a)}:a\in\mathcal A_{i}\}$. Since these
coordinate sets are disjoint across agents, the product $\prod_{i\in\mathcal N_t}p_i(A;x)$ is multilinear in the variables associated with different agents. Hence $\tilde f_t(\cdot\,;s_t)$ is a finite sum of multilinear polynomials and is smooth.

4) By \cref{lem:pme_coordinate_gradient}, for any $(i,a)\in\Omega_t$, $\frac{\partial \tilde f_t}{\partial x_{(i,a)}}(x;s_t) = \mathbb E_{A^{-i}\sim\mathcal D_{t}^{-i}(\cdot\mid x)}
  \left[ F_t\big((i,a)\mid A^{-i};s_t\big) \right]$.

For every realization $A^{-i}$ in the support of $\mathcal D_{t}^{-i}(\cdot\mid x)$, we have $A^{-i}\subseteq A^{-i}\cup\{(i,a)\}$. Since $F_t(\cdot\,;s_t)$ is monotone,
$F_t\big((i,a)\mid A^{-i};s_t\big) = F_t\big(A^{-i}\cup\{(i,a)\};s_t\big)-F_t(A^{-i};s_t) \ge 0 $.
Taking expectation preserves the inequality, and hence
$\frac{\partial \tilde f_t}{\partial x_{(i,a)}}(x;s_t) \ge 0, \ (i,a) \in \Omega_t$. Therefore, $\nabla\tilde f_t(x;s_t)\ge0$ componentwise.

5) We prove the DR property by showing that the gradient is componentwise nonincreasing. For each agent $i$, the PME is affine in the local coordinates
$\{x_{(i,a)}:a\in\mathcal A_{i}\}$, because the per-agent factor $p_i(A;x)$ in \eqref{eq:agent_prob} is affine in these coordinates and no
product contains two coordinates associated with the same agent. Hence, all second partial derivatives with respect to two coordinates of the same agent are zero:
\begin{equation}
\frac{\partial^2 \tilde f_t}
{\partial x_{(i,a)}\partial x_{(i,b)}}(x;s_t)=0,
\ a,b\in\mathcal A_{i}.
\label{eq:same_agent_second_partial}
\end{equation}

Consider two distinct agents $i\ne u$ and actions $a\in\mathcal A_{i}$ and $b\in\mathcal A_{u}$. 
Let $\mathcal D_{t}^{-\{i,u\}}(\cdot\mid x)$ denote the product
distribution induced by $x$ over all agents in
$\mathcal N_t\setminus\{i,u\}$. By~\cref{lem:pme_coordinate_gradient},
\begin{equation}
\frac{\partial \tilde f_t}{\partial x_{(i,a)}}(x;s_t)
=
\mathbb E_{A^{-i}\sim\mathcal D_{t}^{-i}(\cdot\mid x)}
\left[
F_t\big((i,a)\mid A^{-i};s_t\big)
\right].
\label{eq:first_partial_for_second}
\end{equation}
To differentiate this expression with respect to $x_{(u,b)}$, separate the random choice of agent $u$ from the choices of the other agents. For any fixed realization $A^{-\{i,u\}}$ of the agents other than $i$ and $u$, the contribution of agent $u$ to \eqref{eq:first_partial_for_second} is
\begin{align}
&\left(1-\sum_{c\in\mathcal A_{u}}x_{(u,c)}\right)
F_t\big((i,a)\mid A^{-\{i,u\}};s_t\big)
\nonumber\\
&\quad
+\sum_{c\in\mathcal A_{u}}x_{(u,c)}
F_t\big((i,a)\mid A^{-\{i,u\}}\cup\{(u,c)\};s_t\big).
\label{eq:agent_u_contribution}
\end{align}
Differentiating \eqref{eq:agent_u_contribution} with respect to $x_{(u,b)}$ gives
\begin{equation}
\begin{aligned}
&\frac{\partial^2 \tilde f_t}
{\partial x_{(u,b)}\partial x_{(i,a)}}(x;s_t)\\
&=
\mathbb E_{A^{-\{i,u\}}\sim\mathcal D_t^{-\{i,u\}}(\cdot\mid x)}
\Big[
F_t\big((i,a)\mid A^{-\{i,u\}}\cup\{(u,b)\};s_t\big)\\
&\qquad\qquad\qquad\qquad
-
F_t\big((i,a)\mid A^{-\{i,u\}};s_t\big)
\Big].
\end{aligned}
\label{eq:cross_second_partial}
\end{equation}
Since $A^{-\{i,u\}}\subseteq A^{-\{i,u\}}\cup\{(u,b)\}$, submodularity of $F_t(\cdot\,;s_t)$ implies
$F_t\big((i,a)\mid A^{-\{i,u\}}\cup\{(u,b)\};s_t\big)
- F_t\big((i,a)\mid A^{-\{i,u\}};s_t\big)
\le 0$.
Taking expectation in \eqref{eq:cross_second_partial} yields
\begin{equation}
\frac{\partial^2 \tilde f_t}
{\partial x_{(u,b)}\partial x_{(i,a)}}(x;s_t)
\le 0,
\ i\ne u.
\label{eq:cross_second_partial_nonpositive}
\end{equation}
Equations \eqref{eq:same_agent_second_partial} and \eqref{eq:cross_second_partial_nonpositive} show that each coordinate derivative of $\tilde f_t$ is nonincreasing in every coordinate. Therefore, for any $x,y\in\mathcal P(\mathcal I_t)$ with $x\le y$ componentwise,
$\nabla \tilde f_t(x;s_t)\ge \nabla \tilde f_t(y;s_t)$
componentwise. This establishes DR-submodularity on $\mathcal P(\mathcal I_t)$.
\end{proof}
\cref{pro:pme_properties} shows that the PME is a smooth, monotone, DR-submodular objective on $\mathcal P(\mathcal I_t)$ while preserving the per-agent categorical sampling structure required for decentralized execution~\cite{bian2017guaranteed,hassani2017gradient}.

\subsection{Problem Equivalence}
\label{sec:pme_equivalence}
In this section,
we relate the PME to \cref{prob:main}.
First, we relate the continuous PME optimum to the discrete stagewise
submodular optimum.
Then,
we extend the equivalence to the finite-horizon objective $J(\pi)$.

Although $\tilde f_t(\cdot;s_t)$ is defined over the
partition-matroid polytope $\mathcal P(\mathcal I_t)$, monotonicity
implies that an optimal solution can be chosen on the
categorical-policy face $\mathcal F_t$.

\begin{lemma}[PME Exactness on the Categorical Face]
\label{lem:pme_exactness}
Under \cref{ass:standing_conditions}, the following stagewise
identity holds at each time step $t$ and state $s_t$:
\begin{equation}
\label{eq:pme_exactness}
\max_{x\in\mathcal P(\mathcal I_t)}
\tilde f_t(x;s_t)
=
\max_{x\in\mathcal F_t}
\tilde f_t(x;s_t)
=
\max_{A\in\mathcal I_t}
F_t(A;s_t).
\end{equation}

\end{lemma}

\begin{proof}
We first show that the optimum of $f_t(\cdot;s_t)$ over $\mathcal P(\mathcal I_t)$ is attained on $\mathcal F_t$. Let $x^\star\in\arg\max_{x\in\mathcal P(\mathcal I_t)}\tilde f_t(x;s_t)$. If $x^\star\in\mathcal F_t$, the claim is immediate. Otherwise, there exists an agent $i\in\mathcal N_t$ for which the corresponding partition constraint is nonactive:
$ \rho_i =  1-\sum_{a\in\mathcal A_{i}}x^\star_{(i,a)} > 0$.
Choose any action $a_i\in\mathcal A_{i}$ and define $x'=x^\star+\rho_i e_{(i,a_i)}$, where $e_{(i,a_i)}$ is the canonical vector with nonzero coordinate $(i,a_i)$. Then $x'\in\mathcal P(\mathcal I_t)$ and the $i$-th partition constraint becomes tight. By monotonicity of $\tilde f_t$ from \cref{pro:pme_properties}, $\tilde f_t(x';s_t)\ge\tilde f_t(x^\star;s_t)$. 
Repeating this operation for every agent $i$ such that $\sum_{a\in\mathcal A_{i}}x^\star_{(i,a)}<1$ yields a point $\bar x\in\mathcal F_t$ with $\tilde f_t(\bar x;s_t)\ge\tilde f_t(x^\star;s_t)$. Since $x^\star$ is optimal over $\mathcal P(\mathcal I_t)$, equality holds. Thus, $\max_{x\in\mathcal P(\mathcal I_t)}
  \tilde f_t(x;s_t) = \max_{x\in\mathcal F_t}
  \tilde f_t(x;s_t)$.

It remains to relate the optimum over $\mathcal F_t$ to the discrete optimum over $\mathcal I_t$. The set $\mathcal F_t$ is a product of local categorical simplices. Fix the local marginals of all agents except $i$. Then the PME is affine in the local vector $x_i\in\mathcal F_{i}$. Therefore, optimizing $\tilde f_t$ over $x_i$ admits an optimal solution at an extreme point of $\mathcal F_{i}$. Applying this argument successively to all active agents, there exists a maximizer of $\tilde f_t$ over $\mathcal F_t$ that is an extreme point of every local categorical simplex. Such a point assigns probability one to exactly one action for each active agent and therefore corresponds to a deterministic feasible set in $\mathcal I_t$.

Conversely, let $A^\star\in\arg\max_{A\in\mathcal I_t}F_t(A;s_t)$. If $A^\star$ omits an active agent, monotonicity of $F_t$ allows us to add any action from that agent's local action set without decreasing the utility. Repeating this step gives a feasible set $\bar A^\star\in\mathcal I_t$ that selects exactly one action per active agent and satisfies $F_t(\bar A^\star;s_t)\ge F_t(A^\star;s_t)$. The indicator vector $\mathbf 1_{\bar A^\star}$ lies in $\mathcal F_t$, and by the PME definition at deterministic vertices, $ \tilde f_t(\mathbf 1_{\bar A^\star};s_t) =  F_t(\bar A^\star;s_t)$.

Therefore the continuous optimum over $\mathcal F_t$ and the discrete optimum over $\mathcal I_t$ have the same value.
\end{proof}

\cref{lem:pme_exactness} shows that, for monotone utilities, the PME admits an optimal solution on the categorical-policy face $\mathcal F_t$. Therefore, restricting to factorized categorical policies does not reduce the optimal value.

The stagewise equivalence in \cref{lem:policy_pme_equivalence}
yields an exact reformulation of the finite-horizon objective in
\cref{prob:main}. The following corollary shows that the original
expected cumulative submodular utility is equal to the cumulative PME
objective evaluated along the trajectory induced by the factorized
categorical policy sequence.

\begin{corollary}[Problem Equivalence]
\label{cor:cumulative_equivalence}
Under the factorized categorical policy sequence
$\pi=\{\pi_t\}_{t=1}^{T}$, the finite-horizon objective in
\cref{prob:main} is exactly equal to the expected cumulative PME
objective evaluated along the induced state trajectory:
\begin{equation}
\label{eq:pme_policy_objective}
J(\pi)
=
\mathbb E_{\tau\sim p^\pi}
\left[
\sum_{t=1}^{T}
\tilde f_t\bigl(x_t^\pi(s_t);s_t\bigr)
\right],
\end{equation}
where $p^\pi$ is the trajectory distribution induced by $\pi$ and the
transition kernels $\{P_t\}_{t=1}^{T}$.
\end{corollary}

\begin{proof}
By the definition of the finite-horizon objective,
$J(\pi)=\mathbb E_{\tau\sim p^\pi}\left[\sum_{t=1}^{T}
F_t(A_t;s_t)\right]$,
where the trajectory distribution $p^\pi$ is induced by the policy
sequence $\pi$ and the transition kernels $\{P_t\}_{t=1}^{T}$. By
linearity of expectation,
$J(\pi)=\sum_{t=1}^{T}\mathbb E_{\tau\sim p^\pi}\left[
F_t(A_t;s_t)\right]$.
For each time step $t$, applying the tower property of conditional
expectation with respect to the state $s_t$ gives
$\mathbb E_{\tau\sim p^\pi}
\left[
F_t(A_t;s_t)
\right]
=
\mathbb E_{\tau\sim p^\pi}
\left[
\mathbb E_{A_t\sim\pi_t}
\left[
F_t(A_t;s_t)
\mid s_t
\right]
\right]$.
Conditional on $s_t$, the local observations
$\{o_{i,t}\}_{i\in\mathcal N_t}$ are fixed, and the action set $A_t$ is
sampled from the factorized categorical policy $\pi_t$. Therefore, by
\cref{lem:policy_pme_equivalence},
$\mathbb E_{A_t\sim\pi_t}
\left[
F_t(A_t;s_t)
\mid s_t
\right]
=\tilde f_t\bigl(x_t^\pi(s_t);s_t\bigr)$.
Substituting this identity into the preceding display yields
$\mathbb E_{\tau\sim p^\pi}
\left[
F_t(A_t;s_t)
\right]
=
\mathbb E_{\tau\sim p^\pi}
\left[
\tilde f_t\bigl(x_t^\pi(s_t);s_t\bigr)
\right]$.
Summing over $t=1,\ldots,T$ gives
\begin{equation*}
J(\pi)
=
\mathbb E_{\tau\sim p^\pi}
\left[
\sum_{t=1}^{T}
\tilde f_t\bigl(x_t^\pi(s_t);s_t\bigr)
\right],
\end{equation*}
which proves \eqref{eq:pme_policy_objective}.
\end{proof}


\section{Submodular Multi-Agent Policy Learning}
\label{sec:algorithm}

This section describes our policy-learning method, Submodular Multi-Agent Policy Learning (SubMAPL). By \cref{cor:cumulative_equivalence}, policy learning for \cref{prob:main} can be carried out through the PME objectives evaluated along the induced state trajectory. The main challenge is to construct tractable first-order feedback and update the local categorical policies without violating their simplex constraints.

Mirror ascent is a standard first-order method for constrained optimization that replaces Euclidean projection with a Bregman-divergence regularization \cite{beck2003mirror,hazan2016introduction}. For categorical distributions, the KL divergence is particularly suitable because the resulting mirror step admits a multiplicative update that preserves nonnegativity and normalization without Euclidean projection
\cite{shahrampour2017distributed,zhang2025near}. Motivated by this geometry, and in contrast to the Euclidean projected-gradient and policy-gradient framework studied in our preliminary work
\cite{liu2026submodular}, SubMAPL uses this PME representation to construct local marginal-gain feedback during training and updates local categorical policies through a KL-mirror step.

SubMAPL follows the centralized training-decentralized execution (CTDE) paradigm \cite{lowe2017multi,foerster2018coma}. In the training phase, the stage utility $F_t(\cdot\,;s_t)$ is evaluated on feasible sets that differ only in one agent's action to construct the stochastic local gradient estimator used in the KL-mirror update. In execution, each active agent applies its learned local policy using observations $o_{i,t}$.  

\subsection{Tabular Categorical Policies}
\label{sec:tabular_policies}

For each active agent $i\in\mathcal N_t$, SubMAPL uses a tabular softmax parameterization of the local categorical policy. Here, ``tabular'' means that a separate trainable logit is maintained for each local observation-action entry~\cite{sutton1999policy,peters2008natural,agarwal2021theory}. Let $\theta^i(o,a)\in\mathbb R$ denote the logit of agent $i$ for observation $o \in\mathcal O_i$ and action $a\in\mathcal A_{i}$. Given the local observation $o_{i,t}$, the policy is
\begin{equation}
\label{eq:tabular_softmax_main}
 \pi^i(a\mid o_{i,t})
 =
 \frac{
 \exp\bigl(\theta^i(o_{i,t},a)\bigr)
 }{
 \sum_{b\in\mathcal A_{i}}
 \exp\bigl(\theta^i(o_{i,t},b)\bigr)
 },
 \ a\in\mathcal A_{i}.
\end{equation}
The logits $\{\theta^i(o,a)\}_{o\in\mathcal O_i,\,
a\in\mathcal A_{i}}$ parameterize the tabular policy $\pi^i$. Equation~\eqref{eq:tabular_softmax_main} evaluates this policy at the current observation $o_{i,t}$. The induced marginal vector satisfies
$x_{(i,a),t}^{\pi}(s_t) = \pi^i\bigl(a\mid o_{i,t}\bigr), \ (i,a)\in\Omega_t$.

In open systems, the active-agent set $\mathcal N_t$ may vary over time. The logits of remaining agents are inherited across consecutive open-system domains. Each arriving agent is initialized using finite logits. Departing agents are excluded from the active policy domain. Hence, every feasible action of each active agent has strictly positive probability under the local softmax policy.

\subsection{Marginal-Gain Feedback and KL-Mirror Update}

At time $t$, condition on the current state $s_t$. For each active agent $i\in\mathcal N_t$ and $a\in\mathcal A_{i}$, define the local PME gradient coordinate by
\begin{equation}
\label{eq:true_pme_gradient}
g_{i,a,t}(\pi_t;s_t)
=\frac{\partial\tilde f_t}{\partial x_{(i,a)}}
(x_t^{\pi}(s_t); s_t),\ (i,a)\in\Omega_t.
\end{equation}
By \cref{lem:pme_coordinate_gradient}, this coordinate admits the marginal-gain representation
$g_{i,a,t}(\pi_t;s_t) = \mathbb E\left[F_t((i,a)\mid A^{-i};s_t)
\right]$,
where the actions in $A^{-i}$ are sampled independently according to
$\pi^j(\cdot\mid o_{j,t})$ for $j\in\mathcal N_t\setminus\{i\}$.

After the joint action
$A_t=\{(i,a_{i,t}):i\in\mathcal N_t\}$ is sampled, let
$A_t^{-i}=A_t\setminus\{(i,a_{i,t})\}$. SubMAPL constructs the local
marginal-gain estimator
\begin{equation}
\widehat g_{i,a,t} = F_t((i,a)\mid A_t^{-i};s_t),
\ a\in\mathcal A_{i}.
\label{eq:local_gradient_estimator}
\end{equation}

The vector $\widehat g_{i,\cdot,t}$ evaluates every feasible action of agent $i$ under the same sampled actions of the other agents. Conditional on $\pi_t$ and $s_t$, the actions in $A_t^{-i}$ are sampled independently according to $\pi^j(\cdot\mid o_{j,t})$ for
$j\in\mathcal N_t\setminus\{i\}$. Therefore,
\begin{equation}
\label{eq:unbiased_local_gradient_estimator}
\mathbb E
\left[
\widehat g_{i,a,t}
\mid \pi_t,s_t
\right]
=g_{i,a,t}(\pi_t;s_t) =\frac{\partial \tilde f_t}
{\partial x_{(i,a)}}(x_t^{\pi}(s_t);s_t).
\end{equation}
Thus, $\widehat g_{i,\cdot,t}$ is an unbiased stochastic estimator of the local PME gradient evaluated at the current policy.

The estimator in \eqref{eq:local_gradient_estimator} provides first-order information with respect to the local action probabilities. A direct additive probability update would require projection onto the
probability simplex~\cite{beck2003mirror,hazan2016introduction}. SubMAPL instead applies KL-regularized mirror ascent, which preserves
the simplex constraint without Euclidean projection and can be implemented through tabular softmax logits.

Let $\pi^{i,+}$ denote the tabular policy of agent $i$ after the KL-mirror update at time $t$ and before open-system migration. Given the current local policy
$\pi^i(\cdot\mid o_{i,t})$ and the stochastic local gradient estimator $\widehat g_{i,\cdot,t}$, the local policy distribution at the current observation is updated as \cite{beck2003mirror}
\begin{equation}
\label{eq:kl_mirror_optimization}
\pi^{i,+}(\cdot\mid o_{i,t})
\in
\arg\max_{p\in\mathcal F_{i}}
\left\{
\eta\langle\widehat g_{i,\cdot,t},p\rangle
-
D_{\rm KL}
\bigl(
p\Vert\pi^i(\cdot\mid o_{i,t})
\bigr)
\right\},
\end{equation}
where $\eta>0$ is the step size. For two local categorical distributions
$p,q\in\mathcal F_{i}$, define
\begin{equation}
\label{eq:local_kl_divergence}
D_{\rm KL}(p\Vert q)
=
\sum_{a\in\mathcal A_{i}}
p_a
\log
\frac{p_a}{q_a},
\end{equation}
with the convention $0\log(0/q_a)=0$.
For two tabular policies $\pi$ and $\pi'$ defined on the current
open-system domain, define
\begin{equation}
\label{eq:tabular_policy_kl}
D_{{\rm KL},t}(\pi\Vert\pi')
=
\sum_{i\in\mathcal N_t}
\sum_{o\in\mathcal O_i}
D_{\rm KL}
\left(
\pi^i(\cdot\mid o)
\Vert
\pi'^i(\cdot\mid o)
\right).
\end{equation}

The first-order optimality condition yields
\begin{equation}
\label{eq:kl_mirror_closed_form}
\pi^{i,+}(a\mid o_{i,t})
=\frac{\pi^i(a\mid o_{i,t})
\exp(\eta\widehat g_{i,a,t})
}{
\sum_{b\in\mathcal A_{i}}
\pi^i(b\mid o_{i,t})
\exp(\eta\widehat g_{i,b,t})
},
\
a\in\mathcal A_{i}.
\end{equation}
For all other observations, the local policy distributions remain unchanged:
\begin{equation}
\label{eq:unchanged_policy_distributions}
\pi^{i,+}(\cdot\mid o)
= \pi^i(\cdot\mid o), \
o\in\mathcal O_i\setminus\{o_{i,t}\}.
\end{equation}
Hence, $\pi^{i,+}(\cdot\mid o_{i,t})\in\mathcal F_{i}$.
Since the KL-mirror step changes only the row associated with
$o_{i,t}$, all other terms in the tabular-policy KL divergence remain
unchanged.

Under the tabular softmax parameterization
\eqref{eq:tabular_softmax_main}, \eqref{eq:kl_mirror_closed_form} is implemented by the additive logit update
\begin{equation}
\label{eq:tabular_logit_update}
\theta^{i,+}(o_{i,t},a)
=
\theta^i(o_{i,t},a)
+
\eta\widehat g_{i,a,t},
\
a\in\mathcal A_{i},
\end{equation}
where all other logit entries remain unchanged:
\begin{equation}
\label{eq:unchanged_logit_entries}
\theta^{i,+}(o,a)
=
\theta^i(o,a),
\
o\in\mathcal O_i\setminus\{o_{i,t}\},
\
a\in\mathcal A_{i}.
\end{equation}
Substituting \eqref{eq:tabular_logit_update} into the softmax
parameterization yields
\begin{equation}\label{eq:logit_mirror_equivalence}
\begin{aligned}
\pi^{i,+}(a\mid o_{i,t})
&=\frac{\exp \left(\theta^i(o_{i,t},a)+\eta\widehat g_{i,a,t}
\right)}{
\displaystyle \sum_{b\in\mathcal A_i}
\exp \left(\theta^i(o_{i,t},b)+\eta\widehat g_{i,b,t}\right)}
\\
&=\frac{\exp \left(\theta^i(o_{i,t},a) \right)\exp (\eta\widehat g_{i,a,t})}{\displaystyle \sum_{b\in\mathcal A_i}
\left( \exp \left(\theta^i(o_{i,t},b)\right)\exp (\eta\widehat g_{i,b,t}) \right)}
\\
&\stackrel{(i)}{=}\frac{\pi^i(a\mid o_{i,t})\exp(\eta\widehat g_{i,a,t})
}{\displaystyle \sum_{b\in\mathcal A_i}\pi^i(b\mid o_{i,t})
\exp(\eta\widehat g_{i,b,t})}, 
\ a\in\mathcal A_i,
\end{aligned}
\end{equation}
where $(i)$ follows by dividing the numerator and denominator by 
$\sum_{b\in\mathcal A_i}\exp(\theta^i(o_{i,t},b))$.
Thus, the additive logit update exactly implements the KL-mirror update in \eqref{eq:kl_mirror_closed_form}.
The post-update active joint policy before migration is defined by
\begin{equation}
\label{eq:post_update_joint_policy}
\pi_t^+(\mathbf a_t\mid\mathbf o_t)
=
\prod_{i\in\mathcal N_t}
\pi^{i,+}(a_i\mid o_{i,t}).
\end{equation}
Therefore, SubMAPL performs a local full-action KL-mirror step on the current local observation of each active agent's tabular policy and implements this policy-space update through an additive logit update.  

\subsection{Open Policy Migration}
\label{sec:policy_migration}
The transition from time step $t$ to $t+1$ involves the remaining agents $\mathcal R_{t+1}$, the arriving agents
$\mathcal N_{t+1}^{\rm in}$, and the departing agents
$\mathcal N_{t+1}^{\rm out}$. After the KL-mirror update at time $t$, let $\pi_t^+$ denote the post-update active joint policy on the current agent domain $\mathcal N_t$. Open policy migration retains the updated local policies of the remaining agents, assigns policies to the arriving agents, and excludes the departing agents from the next active
policy domain.

\begin{definition}[Open Policy Migration]
\label{def:migration}
Given the post-update active joint policy $\pi_t^+$, the open policy migration map $\mathcal M_t$ constructs
$\pi_{t+1}=\mathcal M_t(\pi_t^+)$
according to
\begin{equation}
\label{eq:policy_migration_agentwise}
(\pi_{t+1})^i(\cdot\mid o)
=
\begin{cases}
\pi^{i,+}(\cdot\mid o),
& i\in\mathcal R_{t+1},\\[2mm]
\pi^{i,0}(\cdot\mid o),
& i\in\mathcal N_{t+1}^{\rm in},
\end{cases}
\end{equation}
where $\pi^{i,0}$ is the categorical policy assigned to an arriving agent $i$. It is represented by finite logits
$\theta^{i,0} \in \mathbb R^{|\mathcal O_i|\times|\mathcal A_i|}$, so that $\pi^{i,0}(a\mid o)>0,
o\in\mathcal O_i, a\in\mathcal A_i$.
Agents in $\mathcal N_{t+1}^{\rm out}$ do not appear in the next active joint-policy domain.
\end{definition}

\begin{algorithm}[t]
\small
\caption{Submodular Multi-Agent Policy Learning (SubMAPL)}
\label{alg:submapl}
\begin{algorithmic}[1]
\State {\bfseries Input:} horizon $T$, step size $\eta$, initial active parameter collection
$\Theta_1=\{\theta^{i,0}:i\in\mathcal N_1\}$, finite-logit migration rule for arriving agents, and
stage-utility evaluators $\{F_t\}_{t=1}^T$
\For{$t=1,\ldots,T$}
    \For{each active agent $i\in\mathcal N_t$}
        \State Receive the local observation $o_{i,t}$
        \State Form $\pi^i(\cdot\mid o_{i,t})$ using
        \eqref{eq:tabular_softmax_main}
        \State Sample
        $a_{i,t}\sim\pi^i(\cdot\mid o_{i,t})$
    \EndFor
    \State Form
    $A_t=\{(i,a_{i,t}):i\in\mathcal N_t\}$
    \For{each active agent $i\in\mathcal N_t$}
        \State Set
        $A_t^{-i}=A_t\setminus\{(i,a_{i,t})\}$
        \For{each feasible action $a\in\mathcal A_{i}$}
            \State Evaluate
            $\widehat g_{i,a,t}$ using
            \eqref{eq:local_gradient_estimator} at the pre-transition state $s_t$
           \State Update $\theta^i(o_{i,t}, a) \leftarrow \theta^i(o_{i,t}, a) + \eta \widehat g_{i,a,t}$   
        \EndFor
    \EndFor
    \State Execute $A_t$
    \State Sample the next state
        $s_{t+1}\sim
        P_t(\cdot\mid s_t,\mathbf a_t)$
    \If{$t<T$}   
        \State Observe the next active-agent set
        $\mathcal N_{t+1}$
        \State Apply the open policy migration in \cref{def:migration}
    \EndIf
\EndFor
\State {\bfseries Output:} the online policy sequence
$\{\pi_t\}_{t=1}^T$ and the final post-update policy
$\pi_T^+$
\end{algorithmic}
\end{algorithm}

The complete training procedure is summarized in \cref{alg:submapl}. At each time step, SubMAPL evaluates the marginal gain of every feasible action of each active agent while keeping the sampled actions of the other agents fixed. The learned policy remains decentralized at execution time, because each active agent samples from its own categorical policy using only its local observation.


\section{Performance Analysis}
\label{sec:theory}

This section establishes performance guarantees for SubMAPL as a PME-based KL-mirror policy-learning method in OMAS. The analysis is conducted along the state sequence generated by \cref{alg:submapl}. At each time step, the current state and open-system domain define a conditional stagewise PME objective over the current categorical-policy face. We first collect the basic gradient and mirror-ascent inequalities used in the analysis, then derive a bound on the finite-horizon objective in \cref{prob:main} through regret analysis, and finally specialize the general bound to a tighter one when the finite-horizon objective $J(\pi)$ is itself submodular.

\subsection{Analysis of KL-Mirror and Stagewise Bound}

\begin{assumption}[Bounded Marginal Gains~\cite{zhang2019online}]
\label{ass:bounded_marginal_gain}
There exists a constant $B<\infty$ such that, for every $t\in[T]$ and state $s_t$,
\begin{equation}
\label{eq:bounded_marginal_gain}
  0 \le F_t(e\mid A;s_t) \le B, \ \forall A\subseteq\Omega_t, \ e\in\Omega_t\setminus A.
\end{equation}
\end{assumption}

Let $\widehat g_t=(\widehat g_{i,a,t})_{(i,a)\in\Omega_t}
\in\mathbb R^{|\Omega_t|}$ denote the stochastic marginal-gain vector used by SubMAPL, where $\widehat g_{i,a,t}$ is defined in
\eqref{eq:local_gradient_estimator}. 
Let $\mathcal H_t$ denote the information available before action sampling at time $t$, including the past trajectory, the current state $s_t$, the current open-system domain $(\mathcal N_t,\{\mathcal A_{i}\}_{i\in\mathcal N_t})$, and the current tabular policy $\pi_t$. Since $x_t=x_t^\pi(s_t)$, conditional on $\mathcal H_t$, the quantities $\pi_t$, $x_t$, $s_t$, and the current open-system domain are fixed. By \eqref{eq:unbiased_local_gradient_estimator}, we have
\begin{equation}
\label{eq:submapl_coordinate_unbiased_pme_gradient}
\mathbb E
\left[
\widehat g_{i,a,t}
\mid
\mathcal H_t
\right]
=
\frac{\partial \tilde f_t}{\partial x_{(i,a)}}(x_t;s_t),
\
(i,a)\in\Omega_t.
\end{equation}
Equivalently,
\begin{equation}
\label{eq:submapl_unbiased_pme_gradient}
\mathbb E
\left[
\widehat g_t
\mid
\mathcal H_t
\right]
=
\nabla \tilde f_t(x_t;s_t).
\end{equation}

\cref{ass:bounded_marginal_gain} also gives a uniform bound on the stochastic feedback. Indeed, for every $(i,a)\in\Omega_t$, we
have $(i,a)\notin A_t^{-i}$. Hence \cref{ass:bounded_marginal_gain}
can be applied with $e=(i,a)$ and $A=A_t^{-i}$, which yields $ 0 \le \widehat g_{i,a,t} = F_t((i,a)\mid A_t^{-i};s_t)  \le B, \ (i,a)\in\Omega_t$ .
Consequently, for each active agent $i\in\mathcal N_t$, $\max_{a\in\mathcal A_{i}} |\widehat g_{i,a,t}| \le B $.

For any vector $g\in\mathbb R^{|\Omega_t|}$ indexed by the current
agent-action ground set $\Omega_t$, define the per-agent mixed norm by
\begin{equation}
\label{eq:agentwise_mixed_norm}
  \lVert g\rVert_{\infty,2}^{2}
  =
  \sum_{i\in\mathcal N_t}
  \left(
  \max_{a\in\mathcal A_{i}} |g_{i,a}|
  \right)^2 .
\end{equation}
The norm $\lVert\cdot\rVert_{\infty,2}$ is always understood with
respect to the current open-system domain
$(\mathcal N_t,\{\mathcal A_{i}\}_{i\in\mathcal N_t})$.
Then
\begin{align}
  \lVert \widehat g_t\rVert_{\infty,2}^{2}
  =
  \sum_{i\in\mathcal N_t}
  \left(
  \max_{a\in\mathcal A_{i}}
  |\widehat g_{i,a,t}|
  \right)^2 
  \le
  \sum_{i\in\mathcal N_t}
  B^2 = |\mathcal N_t|B^2 .
\label{eq:submapl_bounded_per_agent_norm}
\end{align}
Let $N_{\max}=\max_{t\in[T]}|\mathcal N_t|$. Then
\begin{equation}
\label{eq:submapl_uniform_gradient_bound}
  \lVert \widehat g_t\rVert_{\infty,2}^{2}
  \le
  N_{\max}B^2,
  \ t\in[T].
\end{equation}

\begin{lemma}[Auxiliary Inequalities]
\label{lem:pme_mirror_basic}
Consider a time step $t$, a state $s_t$, and the open-system domain at time $t$. Let $\tilde f_t(\cdot\,;s_t)$ be the PME of a normalized, monotone, submodular utility $F_t(\cdot\,;s_t)$. Then the following inequalities hold.
\begin{enumerate}
    \item \textbf{Categorical-face first-order bound \cite{bian2017guaranteed}:}
    For any $x,y\in\mathcal F_t$,
    \begin{equation}
    \label{eq:restricted_dr_ineq}
      \frac{1}{2}\tilde f_t(y;s_t)-\tilde f_t(x;s_t)
      \le
      \frac{1}{2}
      \left\langle
      \nabla \tilde f_t(x;s_t),y-x
      \right\rangle .
    \end{equation}

    \item \textbf{KL mirror-ascent inequality \cite{beck2003mirror}:}
    Suppose $x\in\mathcal F_t$ satisfies
    $x_{(i,a)}>0$ for all $i\in\mathcal N_t$ and
    $a\in\mathcal A_{i}$. Let $x^+\in\mathcal F_t$ be generated by
    \begin{equation}
    \label{eq:global_kl_mirror_step}
      x^+
     \in
      \arg\max_{z\in\mathcal F_t}
      \left\{
      \eta\langle g,z\rangle
      -
      D_{\rm KL}(z\Vert x)
      \right\},
    \end{equation}
    where $\eta>0$ and
    \[
      D_{\rm KL}(z\Vert x)
      =
      \sum_{i\in\mathcal N_t}
      \sum_{a\in\mathcal A_{i}}
      z_{(i,a)}
      \log\frac{z_{(i,a)}}{x_{(i,a)}} .
    \]
    Then, for any $u\in\mathcal F_t$,
    \begin{equation}
    \label{eq:kl_mirror_ineq}
      \langle g,u-x\rangle
      \le
      \frac{
      D_{\rm KL}(u\Vert x)
      -
      D_{\rm KL}(u\Vert x^+)
      }{\eta}
      +
      \frac{\eta}{2}
      \lVert g\rVert_{\infty,2}^{2}.
    \end{equation}
\end{enumerate}
\end{lemma}

\begin{proof}
(1) Let $A\sim\mathcal D_t(\cdot\mid x)$ and $Y\sim\mathcal D_t(\cdot\mid y)$ be sampled independently. Since $x,y\in\mathcal F_t$, both $A$ and $Y$ select
exactly one agent-action pair for each active agent. Define
$A=\{e_i:i\in\mathcal N_t\}$ and
$Y=\{u_i:i\in\mathcal N_t\}$, where
$e_i,u_i\in\Omega_{i,t}$. Let $A^{-i}=A\setminus\{e_i\}$.
By monotonicity,
  $F_t(Y;s_t)-F_t(A;s_t)
  \le
  F_t(A\cup Y;s_t)-F_t(A;s_t)$.

Fix an arbitrary ordering $\{i_1,\ldots,i_m\}$ of $\mathcal N_t$,
where $m=|\mathcal N_t|$. By telescoping, $F_t(A\cup Y;s_t)-F_t(A;s_t)  = \sum_{\ell=1}^{m}
  F_t\left(
  u_{i_\ell}
  \mid
  A\cup\{u_{i_1},\ldots,u_{i_{\ell-1}}\};s_t
  \right)$.
Since
  $A^{-i_\ell}
  \subseteq
  A\cup\{u_{i_1},\ldots,u_{i_{\ell-1}}\}$,
submodularity implies
  $F_t\left(
  u_{i_\ell}
  \mid
  A\cup\{u_{i_1},\ldots,u_{i_{\ell-1}}\};s_t
  \right)
  \le
  F_t(u_{i_\ell}\mid A^{-i_\ell};s_t)$.
Therefore,
\begin{equation}
\label{eq:y_gap_upper}
  F_t(Y;s_t)-F_t(A;s_t)
  \le
  \sum_{i\in\mathcal N_t}
  F_t(u_i\mid A^{-i};s_t).
\end{equation}

Next, by normalization and telescoping along the same ordering,
 $ F_t(A;s_t)
  =
  \sum_{\ell=1}^{m}
  F_t\left(
  e_{i_\ell}
  \mid
  \{e_{i_1},\ldots,e_{i_{\ell-1}}\};s_t
  \right)$.
Since
  $\{e_{i_1},\ldots,e_{i_{\ell-1}}\}
  \subseteq
  A^{-i_\ell}$,
submodularity gives
  $F_t\left(
  e_{i_\ell}
  \mid
  \{e_{i_1},\ldots,e_{i_{\ell-1}}\};s_t
  \right)
  \ge
  F_t(e_{i_\ell}\mid A^{-i_\ell};s_t)$.
Hence,
\begin{equation}
\label{eq:a_marginal_lower}
  \sum_{i\in\mathcal N_t}
  F_t(e_i\mid A^{-i};s_t)
  \le
  F_t(A;s_t).
\end{equation}
Combining \eqref{eq:y_gap_upper} and \eqref{eq:a_marginal_lower} yields
\begin{equation}
\label{eq:deterministic_dr_gap}
\begin{aligned}
  & F_t(Y;s_t)-2F_t(A;s_t)\\
  &\le
  \sum_{i\in\mathcal N_t}
  \Big[
  F_t(u_i\mid A^{-i};s_t)
  -
  F_t(e_i\mid A^{-i};s_t)
  \Big].
\end{aligned}
\end{equation}

Taking expectations in \eqref{eq:deterministic_dr_gap} over the independent samples $A\sim\mathcal D_t(\cdot\mid x)$ and $Y\sim\mathcal D_t(\cdot\mid y)$ gives
\[
\begin{aligned}
  & \tilde f_t(y;s_t)-2\tilde f_t(x;s_t)\\
  &\le
  \sum_{i\in\mathcal N_t}
  \sum_{a\in\mathcal A_{i}}
  y_{(i,a)}
  \mathbb E_{A^{-i}\sim\mathcal D_t^{-i}(\cdot\mid x)}
  \left[
  F_t((i,a)\mid A^{-i};s_t)
  \right] \\
  &\quad -
  \sum_{i\in\mathcal N_t}
  \sum_{a\in\mathcal A_{i}}
  x_{(i,a)}
  \mathbb E_{A^{-i}\sim\mathcal D_t^{-i}(\cdot\mid x)}
  \left[
  F_t((i,a)\mid A^{-i};s_t)
  \right].
\end{aligned}
\]
Here, under $\mathcal D_t(\cdot\mid x)$, the selected action of agent
$i$ is independent of $A^{-i}$ and has marginal distribution
$(x_{(i,a)})_{a\in\mathcal A_{i}}$; similarly, under $\mathcal D_t(\cdot\mid y)$, the selected action of agent $i$ has marginal distribution
$(y_{(i,a)})_{a\in\mathcal A_{i}}$ and is independent of $A$.
By \cref{lem:pme_coordinate_gradient}, the expected marginal gain
is the corresponding PME coordinate derivative. Therefore,
\[
  \tilde f_t(y;s_t)-2\tilde f_t(x;s_t)
  \le
  \left\langle
  \nabla \tilde f_t(x;s_t),y-x
  \right\rangle .
\]
Dividing by $2$ proves \eqref{eq:restricted_dr_ineq}.

(2) The optimality condition of \eqref{eq:global_kl_mirror_step} gives
the three-point inequality
\begin{equation}
\label{eq:three_point_kl}
  \eta\langle g,u-x^+\rangle
  \le
  D_{\rm KL}(u\Vert x)
  -
  D_{\rm KL}(u\Vert x^+)
  -
  D_{\rm KL}(x^+\Vert x).
\end{equation}
For each active agent $i\in\mathcal N_t$, H\"older's inequality gives \cite{beck2003mirror,hazan2016introduction}
\begin{equation}
\begin{aligned}
  \eta
 & \sum_{a\in\mathcal A_{i}}
  g_{i,a}
  \left(x^+_{(i,a)}-x_{(i,a)}\right)\\
  &\le
  \eta
  \left(
  \max_{a\in\mathcal A_{i}} |g_{i,a}|
  \right)
  \sum_{a\in\mathcal A_{i}}
  \left|x^+_{(i,a)}-x_{(i,a)}\right|.
\end{aligned}
\end{equation}
By Young's inequality,
\[
\begin{aligned}
  &\eta
  \left(
  \max_{a\in\mathcal A_{i}} |g_{i,a}|
  \right)
  \sum_{a\in\mathcal A_{i}}
  \left|x^+_{(i,a)}-x_{(i,a)}\right| \\
  &\le
  \frac{\eta^2}{2}
  \left(
  \max_{a\in\mathcal A_{i}} |g_{i,a}|
  \right)^2
  +
  \frac{1}{2}
  \left(
  \sum_{a\in\mathcal A_{i}}
  \left|x^+_{(i,a)}-x_{(i,a)}\right|
  \right)^2 .
\end{aligned}
\]
Pinsker's inequality, applied to the two categorical distributions over
$\mathcal A_{i}$, gives
\[
  \frac{1}{2}
  \left(
  \sum_{a\in\mathcal A_{i}}
  \left|x^+_{(i,a)}-x_{(i,a)}\right|
  \right)^2
  \le
  \sum_{a\in\mathcal A_{i}}
  x^+_{(i,a)}
  \log
  \frac{x^+_{(i,a)}}{x_{(i,a)}} .
\]
Thus,
\[
\begin{aligned}
  &\eta
  \sum_{a\in\mathcal A_{i}}
  g_{i,a}
  \left(x^+_{(i,a)}-x_{(i,a)}\right) \\
  &\le
  \frac{\eta^2}{2}
  \left(
  \max_{a\in\mathcal A_{i}} |g_{i,a}|
  \right)^2
  +
  \sum_{a\in\mathcal A_{i}}
  x^+_{(i,a)}
  \log
  \frac{x^+_{(i,a)}}{x_{(i,a)}} .
\end{aligned}
\]
Summing over $i\in\mathcal N_t$ yields
\begin{equation}
\label{eq:mirror_movement_bound}
  \eta\langle g,x^+-x\rangle
  \le
  D_{\rm KL}(x^+\Vert x)
  +
  \frac{\eta^2}{2}
  \lVert g\rVert_{\infty,2}^{2}.
\end{equation}
Combining \eqref{eq:three_point_kl} and
\eqref{eq:mirror_movement_bound}, we obtain
\[
\begin{aligned}
  \eta\langle g,u-x\rangle
  &=
  \eta\langle g,u-x^+\rangle
  +
  \eta\langle g,x^+-x\rangle \\
  &\le
  D_{\rm KL}(u\Vert x)
  -
  D_{\rm KL}(u\Vert x^+)
  +
  \frac{\eta^2}{2}
  \lVert g\rVert_{\infty,2}^{2}.
\end{aligned}
\]
Dividing by $\eta$ proves \eqref{eq:kl_mirror_ineq}.
\end{proof}

\subsection{Suboptimality Gap in the General Case}
\label{sec:regret}

Dynamic approximation regret compares the cumulative utility attained by an algorithm with that of a sequence of stagewise optimal allocations~\cite{zhang2025near,zhang2025effective}. Existing formulations for multi-agent online submodular coordination consider a fixed agent-action domain. However, the consecutive comparator policies in OMAS may belong to different policy domains as agents arrive and depart. We introduce an open-system variant that measures comparator variation after mapping the previous comparator to the new domain through policy migration.

Given the current state $s_t$ and open-system domain, define the stagewise discrete optimum $F_t^\star(s_t)=\max_{A\in\mathcal I_t}F_t(A;s_t)$.
Let $x_t^\star \in \arg\max_{x\in\mathcal F_t} \tilde f_t(x;s_t)$.
By \cref{lem:pme_exactness}, $\tilde f_t(x_t^\star;s_t)=F_t^\star(s_t)$.
The sequence $\{x_t^\star\}_{t=1}^T$ may vary with the state and the active-agent domain.
Define the tabular policy $\pi_t^\star$ associated with $x_t^\star$ by
\begin{equation}
\label{eq:stagewise_optimal_policy}
(\pi_t^{\star})^i(a\mid o)
=
\begin{cases}
x_{(i,a),t}^\star,
& o=o_{i,t},\\
\pi^i(a\mid o),
& o\neq o_{i,t},
\end{cases}
\
i\in\mathcal N_t,\
o\in\mathcal O_i,\
a\in\mathcal A_{i}.
\end{equation}
By \eqref{eq:tabular_policy_kl},
$D_{{\rm KL},t}(\pi_t^\star\Vert\pi_t)=D_{\rm KL}(x_t^\star\Vert x_t)$.

The transition from time step $t$ to $t+1$ consists of remaining agents $\mathcal R_{t+1}$, arriving agents $\mathcal N_{t+1}^{\rm in}$, and departing agents $\mathcal N_{t+1}^{\rm out}$. SubMAPL first obtains the post-update tabular policy $\pi_t^+$ through the KL-mirror update in \eqref{eq:kl_mirror_closed_form}-\eqref{eq:unchanged_policy_distributions}. It then applies the open policy migration in \cref{def:migration}: $\pi_{t+1}=\mathcal M_t(\pi_t^+)$.
Let $\widehat{\pi}=\{\pi_t\}_{t=1}^T$ denote the online policy sequence generated by SubMAPL, and let $p^{\widehat{\pi}}$ denote the trajectory distribution induced by $\widehat{\pi}$ and the transition kernels $\{P_t\}_{t=1}^T$. Following KL-potential tracking analyses on a fixed domain \cite{hall2013dynamical, shahrampour2017distributed}, we define the open-system KL tracking variation as
\begin{equation}
\label{eq:open_kl_variation}
\begin{aligned}
\mathcal V_T^{\rm open}
=
\mathbb E_{\tau\sim p^{\widehat{\pi}}}
\Bigg[
\sum_{t=1}^{T-1}
\Big(
& D_{{\rm KL},t+1}(
\pi_{t+1}^\star\Vert\pi_{t+1})  \\
& -
D_{{\rm KL},t}(
\pi_t^\star\Vert\pi_t^+)
\Big)_+
\Bigg].
\end{aligned}
\end{equation}
This variation records the increase in the KL tracking potential between consecutive time steps while allowing both the stagewise benchmark and the open-system policy domain to change. In particular, it incorporates agent arrivals and departures through $\mathcal R_{t+1}$, $\mathcal N_{t+1}^{\rm in}$, and $\mathcal N_{t+1}^{\rm out}$. When the active-agent set is fixed, $\mathcal M_t$ is the identity map, and $\mathcal V_T^{\rm open}$ reduces to the corresponding KL tracking variation on a fixed tabular-policy domain.

\begin{definition}[Open-System Dynamic Approximation-Regret]
\label{def:open_regret}
The open-system dynamic $1/2$-approximation regret is defined as
\begin{equation}
\label{eq:open_dynamic_alpha_regret}
\mathrm{Reg}_T^{1/2}
=
\sum_{t=1}^{T}
\mathbb E_{\tau\sim p^{\widehat{\pi}}}
\left[
\frac{1}{2}F_t^\star(s_t)
-
\tilde f_t(x_t;s_t)
\right].
\end{equation}
By \cref{lem:policy_pme_equivalence},
$\tilde f_t(x_t;s_t)$ is the conditional expected stage utility of
the factorized categorical policy $\pi_t$. Thus,
$\mathrm{Reg}_T^{1/2}$ measures its cumulative approximation gap
relative to the time-varying stagewise optimal values.
\end{definition}

\begin{lemma}[Open-System KL-Mirror Approximation-Regret Bound]
\label{lem:open_kl_regret} 
Let $\{\pi_t\}_{t=1}^T$ be the policy sequence generated by \cref{alg:submapl}. Under
\cref{ass:standing_conditions} and \cref{ass:bounded_marginal_gain}, for any constant step size
$\eta>0$,
\begin{equation}
\label{eq:open_kl_regret_bound}
\mathrm{Reg}_T^{1/2}
\le
\frac{\mathbb E_{s_1\sim\mu_1}
\left[
D_{{\rm KL},1}
\left(
\pi_1^\star\Vert\pi_1
\right)
\right]
+
\mathcal V_T^{\rm open}
}{2\eta}
+
\frac{\eta T N_{\max}B^2}{4}.
\end{equation}
\end{lemma}

\begin{proof}
For each time step $t$, the categorical-face first-order bound in
\cref{lem:pme_mirror_basic} gives
\begin{equation}
\label{eq:open_restricted_dr_step}
\frac{1}{2}F_t^\star(s_t)
-
\tilde f_t(x_t;s_t)
\le
\frac{1}{2}
\left\langle
\nabla\tilde f_t(x_t;s_t),
x_t^\star-x_t
\right\rangle.
\end{equation}
Since $x_t$ and $x_t^\star$ are measurable with respect to the
pre-sampling history $\mathcal H_t$, taking expectations and using \eqref{eq:submapl_unbiased_pme_gradient} yields
\begin{equation}
\label{eq:open_expected_gap}
\begin{aligned}
&
\mathbb E_{\tau\sim p^{\widehat{\pi}}}
\left[
\frac{1}{2}F_t^\star(s_t)
-
\tilde f_t(x_t;s_t)
\right]
\\
&\quad\le
\frac{1}{2}
\mathbb E_{\tau\sim p^{\widehat{\pi}}}
\left[
\left\langle
\widehat g_t,
x_t^\star-x_t
\right\rangle
\right].
\end{aligned}
\end{equation}

Since $\mathcal F_t=\prod_{i\in\mathcal N_t}\mathcal F_{i}$,
the local KL-mirror updates in \eqref{eq:kl_mirror_optimization} jointly realize the product-space update in  \cref{lem:pme_mirror_basic}. Apply the KL mirror-ascent
inequality with $x=x_t$, $u=x_t^\star$, and
$g=\widehat g_t$, where the post-update vector $x_t^+$ has coordinates
\[
x_{(i,a),t}^+
=
\pi^{i,+}(a\mid o_{i,t}),
\
(i,a)\in\Omega_t.
\]
By \eqref{eq:tabular_policy_kl} and \eqref{eq:stagewise_optimal_policy}, together with the fact that
$\pi_t^+$ and $\pi_t$ coincide at all noncurrent observations,
\begin{equation}
\label{eq:policy_marginal_kl_equivalence}
\begin{aligned}
D_{\rm KL}(x_t^\star\Vert x_t)
&=
D_{{\rm KL},t}
\left(
\pi_t^\star\Vert\pi_t
\right),
\\
D_{\rm KL}(x_t^\star\Vert x_t^+)
&=
D_{{\rm KL},t}
\left(
\pi_t^\star\Vert\pi_t^+
\right).
\end{aligned}
\end{equation}
Therefore,
\begin{equation}
\label{eq:open_mirror_one_step}
\left\langle
\widehat g_t,
x_t^\star-x_t
\right\rangle
\le
\frac{
D_{{\rm KL},t}
\left(
\pi_t^\star\Vert\pi_t
\right)
-
D_{{\rm KL},t}
\left(
\pi_t^\star\Vert\pi_t^+
\right)
}{\eta}
+
\frac{\eta}{2}
\lVert\widehat g_t\rVert_{\infty,2}^2.
\end{equation}
By the uniform gradient bound in \eqref{eq:submapl_uniform_gradient_bound},
\begin{equation}
\label{eq:open_gradient_norm_bound}
\lVert\widehat g_t\rVert_{\infty,2}^2
\le
N_{\max}B^2.
\end{equation}
Combining \eqref{eq:open_expected_gap},
\eqref{eq:open_mirror_one_step}, and
\eqref{eq:open_gradient_norm_bound} gives
\begin{align}
&
\mathbb E_{\tau\sim p^{\widehat{\pi}}}
\left[
\frac{1}{2}F_t^\star(s_t)
-
\tilde f_t(x_t;s_t)
\right]
\notag\\
&\quad\le
\frac{
\mathbb E_{\tau\sim p^{\widehat{\pi}}}
\left[
D_{{\rm KL},t}
\left(
\pi_t^\star\Vert\pi_t
\right)
-
D_{{\rm KL},t}
\left(
\pi_t^\star\Vert\pi_t^+
\right)
\right]
}{2\eta}
+
\frac{\eta N_{\max}B^2}{4}.
\label{eq:open_one_step_bound}
\end{align}

For compactness, define $\Phi_t=D_{{\rm KL},t}\left(\pi_t^\star\Vert \pi_t \right)$,
and $\Gamma_t = D_{{\rm KL},t} \left( \pi_t^\star\Vert\pi_t^+
\right)$.
For $t=1,\ldots,T-1$, we have
\begin{equation}
\label{eq:open_tracking_decomposition}
\begin{aligned}
\Phi_t-\Gamma_t
&=
\Phi_t-\Phi_{t+1}
+
\Phi_{t+1}-\Gamma_t
\\
&\le
\Phi_t-\Phi_{t+1}
+
\left[
\Phi_{t+1}-\Gamma_t
\right]_+.
\end{aligned}
\end{equation}
By the definition of the two potentials,
\begin{equation}
\label{eq:open_potential_increment}
\left[
\Phi_{t+1}-\Gamma_t
\right]_+
=
\left[
D_{{\rm KL},t+1}
\left(
\pi_{t+1}^\star\Vert\pi_{t+1}
\right)
-
D_{{\rm KL},t}
\left(
\pi_t^\star\Vert\pi_t^+
\right)
\right]_+.
\end{equation}
Summing \eqref{eq:open_tracking_decomposition} over
$t=1,\ldots,T-1$ gives
\begin{equation}
\label{eq:open_tracking_partial_sum}
\sum_{t=1}^{T-1}
(\Phi_t-\Gamma_t)
\le
\Phi_1-\Phi_T
+
\sum_{t=1}^{T-1}
\left[
\Phi_{t+1}-\Gamma_t
\right]_+.
\end{equation}
Adding the terminal term $\Phi_T-\Gamma_T$ to both sides and using $\Gamma_T\ge0$ yields
\begin{equation}
\label{eq:open_tracking_pathwise_sum}
\sum_{t=1}^{T}
(\Phi_t-\Gamma_t)
\le
\Phi_1
+
\sum_{t=1}^{T-1}
\left[
\Phi_{t+1}-\Gamma_t
\right]_+.
\end{equation}
Taking expectations with respect to $\tau\sim p^{\widehat{\pi}}$ and using the definition of $\mathcal V_T^{\rm open}$ in \eqref{eq:open_kl_variation}, we obtain
\begin{equation}
\label{eq:open_tracking_sum}
\sum_{t=1}^{T}
\mathbb E_{\tau\sim p^{\widehat{\pi}}}
\left[
\Phi_t-\Gamma_t
\right]
\le
\mathbb E_{\tau\sim p^{\widehat{\pi}}}
\left[
\Phi_1
\right]
+
\mathcal V_T^{\rm open}.
\end{equation}
Summing \eqref{eq:open_one_step_bound} over
$t=1,\ldots,T$ and applying
\eqref{eq:open_tracking_sum} gives
\begin{align}
&
\sum_{t=1}^{T}
\mathbb E_{\tau\sim p^{\widehat{\pi}}}
\left[
\frac{1}{2}F_t^\star(s_t)
-
\tilde f_t(x_t;s_t)
\right]
\notag\\
&\quad\le
\frac{
\mathbb E_{s_1\sim\mu_1}
\left[
D_{{\rm KL},1}
\left(
\pi_1^\star\Vert\pi_1
\right)
\right]
+
\mathcal V_T^{\rm open}
}{2\eta}
+
\frac{\eta T N_{\max}B^2}{4}.
\label{eq:open_regret_sum}
\end{align}
The left-hand side equals $\mathrm{Reg}_T^{1/2}$ by
\cref{def:open_regret}, which proves
\eqref{eq:open_kl_regret_bound}.
\end{proof}

\cref{lem:open_kl_regret} controls the cumulative approximation gap relative to the time-varying stagewise optimal values. However, \cref{prob:main} evaluates the cumulative policy value induced over the full trajectory.  We next connect the stagewise approximation-regret guarantee in \cref{lem:open_kl_regret} to the finite-horizon objective of \cref{prob:main}. By \cref{cor:cumulative_equivalence},
\begin{equation}
\label{eq:submapl_pme_objective}
J(\widehat{\pi})
=
\mathbb E_{\tau\sim p^{\widehat{\pi}}}
\left[
\sum_{t=1}^{T}
\tilde f_t(x_t;s_t)
\right].
\end{equation}
Let $J^\star=\max_{\pi}J(\pi)$ be the optimal finite-horizon value over the considered class of decentralized factorized policy sequences. For readability, denote the bound in \eqref{eq:open_kl_regret_bound} by
\begin{equation}
\label{eq:CT_def}
\mathcal C_T
=
\frac{
\mathbb E_{s_1\sim\mu_1}
\left[
D_{{\rm KL},1}
\left(
\pi_1^\star\Vert\pi_1
\right)
\right]
+
\mathcal V_T^{\rm open}
}{2\eta}
+
\frac{\eta T N_{\max}B^2}{4}.
\end{equation}
Moreover, define the finite-horizon benchmark mismatch
\begin{equation}
\label{eq:horizon_stagewise_gap}
\Delta_T = \left[J^\star - \mathbb E_{\tau\sim p^{\widehat{\pi}}}
\left[\sum_{t=1}^{T} F_t^\star(s_t) \right] \right]_+.
\end{equation}

\begin{theorem}[Finite-Horizon Performance Bound]
\label{thm:finite_horizon_performance}
Let $\widehat{\pi}=\{\pi_t\}_{t=1}^T$ be the policy sequence generated by \cref{alg:submapl}. Under \cref{ass:standing_conditions} and \cref{ass:bounded_marginal_gain}, we have
\begin{equation}
\label{eq:horizon_gap_decomposition}
\frac{1}{2}J^\star-J(\widehat{\pi})
\le
\mathcal C_T+
\frac{1}{2}\Delta_T,
\end{equation}
where $\mathcal C_T$ and $\Delta_T$ are defined in \eqref{eq:CT_def} and \eqref{eq:horizon_stagewise_gap}, respectively.
\end{theorem}

\begin{proof}
By \cref{def:open_regret} and \cref{lem:open_kl_regret},
\begin{equation}
\label{eq:own_stagewise_benchmark_lower_bound}
J(\widehat{\pi})
\ge
\frac{1}{2}
\mathbb E_{\tau\sim p^{\widehat{\pi}}}
\left[
\sum_{t=1}^{T}
F_t^\star(s_t)
\right]
-
\mathcal C_T.
\end{equation}
By the definition of $\Delta_T$ in \eqref{eq:horizon_stagewise_gap},
\begin{equation}
\label{eq:horizon_stagewise_gap_lower_bound}
\mathbb E_{\tau\sim p^{\widehat{\pi}}}
\left[
\sum_{t=1}^{T}
F_t^\star(s_t)
\right]
\ge
J^\star-\Delta_T.
\end{equation}
Substituting \eqref{eq:horizon_stagewise_gap_lower_bound} into \eqref{eq:own_stagewise_benchmark_lower_bound} and rearranging yields \eqref{eq:horizon_gap_decomposition}.
\end{proof}

The mismatch term $\Delta_T$ captures the discrepancy between the optimal finite-horizon value and the cumulative stagewise optimal values along the trajectory induced by SubMAPL. Since the stagewise KL-mirror analysis does not control this term without additional temporal structure, \cref{thm:finite_horizon_performance} provides a finite-horizon decomposition rather than a uniform approximation guarantee.

\subsection{Suboptimality Gap With Submodular Finite-Horizon Utility}
\label{sec:coverage_connection}

We next identify a structural setting in which the stagewise approximation-regret bound yields a finite-horizon approximation guarantee. In coverage and information-collection problems, the stage utility can often be represented as the marginal gain of a
trajectory-level submodular function with respect to the accumulated history~\cite{prajapat2024submodular}. Under such a time-expanded representation, if the ground set and the feasible families are independent of the executed policy, submodularity connects the cumulative stagewise oracle values along the SubMAPL trajectory with the finite-horizon optimum $J^\star$.

\begin{assumption}[Time-Expanded Submodular Utility~\cite{prajapat2024submodular}]
\label{ass:telescoping_coverage}
There exist a policy-independent time-expanded ground set
$\bar\Omega_{1:T}$, policy-independent feasible families
$\{\bar{\mathcal I}_t\}_{t=1}^{T}$, and a normalized, monotone,
submodular function $\bar F:2^{\bar\Omega_{1:T}}\to\mathbb R_{\ge0}$. For every feasible execution, each $A_t\in\mathcal I_t$ corresponds to a time-expanded action set $\bar A_t\in\bar{\mathcal I}_t$, and the state $s_t$ contains the
accumulated history $\bar A_{1:t-1}=\bigcup_{\tau=1}^{t-1}\bar A_\tau$.
The stage utility satisfies
$F_t(A_t;s_t)=\bar F(\bar A_t\mid\bar A_{1:t-1})=
\bar F(\bar A_{1:t-1}\cup\bar A_t)-\bar F(\bar A_{1:t-1})$.
Moreover, every $\bar A_t'\in\bar{\mathcal I}_t$ corresponds to some feasible $A_t'\in\mathcal I_t$.
\end{assumption}

\cref{ass:telescoping_coverage} is used to convert the stagewise approximation-regret bound into a finite-horizon guarantee. It holds for trajectory-level coverage and information-collection objectives whose cumulative utility is represented by a monotone submodular function of the selected time-expanded items~\cite{krause2008near,golovin2011adaptive,prajapat2024submodular}. 

\begin{example}[Static Sensing Coverage Utility~\cite{sun2019exploiting}]
\label{ex:time_expanded_coverage}
Consider a multi-agent coverage problem over a finite target set
$\mathcal Q$. At time step $t$, each agent-action pair $(i,a)\in\Omega_t$ covers a sensing region $C_t(i,a)\subseteq\mathcal Q$, that does not depend on the current state. Assume that $\Omega_t$ and the sensing regions $\{C_t(i,a)\}_{(i,a)\in\Omega_t}$ are determined independently of the executed policy.

Define the time-expanded ground set
$\bar\Omega_{1:T}
=
\{(i,a,t):t\in[T],\ (i,a)\in\Omega_t\}$.
For each feasible $A_t\in\mathcal I_t$, let
$\bar A_t
=
\{(i,a,t):(i,a)\in A_t\}$,
and let
$\bar{\mathcal I}_t
=
\{\bar A_t:A_t\in\mathcal I_t\}$.
Define
$\bar A_{1:t}=\bigcup_{\tau=1}^{t}\bar A_\tau$ and
$\bar A_{1:0}=\emptyset$.
Given target weights $w_q\ge0$, define
\[
\bar F(S)
=
\sum_{q\in\mathcal Q}
w_q
\mathbf 1
\left\{
q\in
\bigcup_{(i,a,t)\in S}
C_t(i,a)
\right\},
\
S\subseteq\bar\Omega_{1:T}.
\]
Then $\bar F$ is normalized, monotone, and submodular. If $s_t$
records the targets covered before time $t$, the stage utility satisfies
$F_t(A_t;s_t)=\bar F(\bar A_t\mid\bar A_{1:t-1})$,
and therefore
\[
\sum_{t=1}^{T}F_t(A_t;s_t)
=
\bar F(\bar A_{1:T}).
\]
Thus, \cref{ass:telescoping_coverage} holds.
\end{example}

\begin{theorem}[Finite-Horizon Approximation under Time-Expanded Submodularity]
\label{thm:telescoping_global}
Let $\widehat{\pi}$ be the policy sequence generated by
\cref{alg:submapl}. Suppose \cref{ass:standing_conditions}, \cref{ass:bounded_marginal_gain}, and \cref{ass:telescoping_coverage} hold. Then
\begin{equation}
\label{eq:telescoping_one_third}
J(\widehat{\pi}) \ge \frac{1}{3}J^\star - \frac{2}{3}\mathcal C_T.
\end{equation}
\end{theorem}

\begin{proof}
Consider an arbitrary trajectory generated by $\widehat{\pi}$. Let
$\{\bar A_t\}_{t=1}^{T}$ denote the corresponding time-expanded action sets, and let $\{s_t\}_{t=1}^{T}$ denote the state sequence. Let $\{\bar A_t'\}_{t=1}^{T}$ be any fixed feasible reference sequence with $\bar A_t'\in\bar{\mathcal I}_t$ for each $t$, and define $\bar A_{1:t}'=\bigcup_{\tau=1}^{t}\bar A_\tau'$ with $\bar A_{1:0}'=\emptyset$. By \cref{ass:telescoping_coverage}, for each
$\bar A_t'\in\bar{\mathcal I}_t$, there exists $A_t'\in\mathcal I_t$ whose corresponding time-expanded action set is $\bar A_t'$.

By monotonicity of $\bar F$,
\begin{equation}
\label{eq:greedy_step1}
\bar F(\bar A_{1:T}')
\le
\bar F(\bar A_{1:T})
+
\bar F(\bar A_{1:T}'\mid \bar A_{1:T}) .
\end{equation}
Using the chain rule for marginal gains, we have
\begin{align}
\label{eq:greedy_step2}
\bar F(\bar A_{1:T}'\mid \bar A_{1:T})
&=
\sum_{t=1}^{T}
\bar F
\left(
\bar A_t'
\mid
\bar A_{1:T}\cup \bar A_{1:t-1}'
\right)
\nonumber\\
&\le
\sum_{t=1}^{T}
\bar F
\left(
\bar A_t'
\mid
\bar A_{1:t-1}
\right),
\end{align}
where the inequality follows from
$\bar A_{1:t-1}\subseteq \bar A_{1:T}\cup\bar A_{1:t-1}'$ and the
diminishing-returns property of the submodular function $\bar F$.

By \cref{ass:telescoping_coverage}, the marginal gain of the reference time-expanded action set $\bar A_t'$ along the SubMAPL history is exactly the corresponding stage utility:
\begin{equation}
\label{eq:comparator_action_set_bound}
\bar F(\bar A_t'\mid \bar A_{1:t-1})
=
F_t(A_t';s_t)
\le
F_t^\star(s_t).
\end{equation}
Combining \eqref{eq:greedy_step1}--\eqref{eq:comparator_action_set_bound}
gives, for the trajectory under consideration,
\begin{equation}
\label{eq:deterministic_comparator_bound}
\bar F(\bar A_{1:T}')
\le
\bar F(\bar A_{1:T})
+
\sum_{t=1}^{T}
F_t^\star(s_t).
\end{equation}

We next take expectation with respect to the trajectory distribution induced by SubMAPL. By \cref{ass:telescoping_coverage}, the
cumulative utility along the SubMAPL trajectory telescopes:
\begin{align*}
\sum_{t=1}^{T}
F_t(A_t;s_t)
&=
\sum_{t=1}^{T}
\bar F(\bar A_t\mid \bar A_{1:t-1})
\\
&=
\sum_{t=1}^{T}
\left[
\bar F(\bar A_{1:t})
-
\bar F(\bar A_{1:t-1})
\right]
\\
&=
\bar F(\bar A_{1:T}),
\end{align*}
where $\bar F(\emptyset)=0$ because $\bar F$ is normalized. Therefore,
\begin{equation*}
J(\widehat{\pi})
=
\mathbb E_{\tau\sim p^{\widehat{\pi}}}
\left[
\bar F(\bar A_{1:T})
\right].
\end{equation*}
Taking expectation in \eqref{eq:deterministic_comparator_bound} over
$\tau\sim p^{\widehat{\pi}}$ gives
\begin{equation}
\label{eq:expected_comparator_bound}
\bar F(\bar A_{1:T}')
\le
J(\widehat{\pi})
+
\mathbb E_{\tau\sim p^{\widehat{\pi}}}
\left[
\sum_{t=1}^{T}
F_t^\star(s_t)
\right].
\end{equation}

Since \eqref{eq:expected_comparator_bound} holds for every feasible
reference sequence $\{\bar A_t'\}_{t=1}^{T}$ with
$\bar A_t'\in\bar{\mathcal I}_t$, it also holds after maximizing the
left-hand side over all such reference sequences:
\begin{equation}
\label{eq:best_reference_bound}
\max_{\bar A_t'\in\bar{\mathcal I}_t,\ t\in[T]}
\bar F(\bar A_{1:T}')
\le
J(\widehat{\pi})
+
\mathbb E_{\tau\sim p^{\widehat{\pi}}}
\left[
\sum_{t=1}^{T}
F_t^\star(s_t)
\right].
\end{equation}
Since the time-expanded feasible families are policy-independent, every realization generated by an admissible policy $\pi$ satisfies
$\bar A_t\in\bar{\mathcal I}_t$ for all $t\in[T]$. Therefore,
for any admissible policy $\pi$,
\begin{equation*}
J(\pi)
=
\mathbb E_{\tau\sim p^\pi}
\left[
\bar F(\bar A_{1:T})
\right]
\le
\max_{\bar A_t'\in\bar{\mathcal I}_t,\ t\in[T]}
\bar F(\bar A_{1:T}').
\end{equation*}
Taking the maximum over $\pi$ gives
\begin{equation*}
J^\star
\le
\max_{\bar A_t'\in\bar{\mathcal I}_t,\ t\in[T]}
\bar F(\bar A_{1:T}') .
\end{equation*}
Combining this inequality with \eqref{eq:best_reference_bound} yields
\begin{equation}
\label{eq:Jstar_leq_J_plus_oracle}
J^\star
\le
J(\widehat{\pi})
+
\mathbb E_{\tau\sim p^{\widehat{\pi}}}
\left[
\sum_{t=1}^{T}
F_t^\star(s_t)
\right].
\end{equation}

It remains to use the stagewise approximation-regret bound. From
\eqref{eq:own_stagewise_benchmark_lower_bound}, we have
\begin{equation}
\label{eq:oracle_upper_submapl}
\mathbb E_{\tau\sim p^{\widehat{\pi}}}
\left[
\sum_{t=1}^{T}
F_t^\star(s_t)
\right]
\le
2J(\widehat{\pi})
+
2\mathcal C_T .
\end{equation}
Substituting \eqref{eq:oracle_upper_submapl} into
\eqref{eq:Jstar_leq_J_plus_oracle} gives
\begin{equation*}
J^\star
\le
3J(\widehat{\pi})
+
2\mathcal C_T .
\end{equation*}
Rearranging proves \eqref{eq:telescoping_one_third}.
\end{proof}

\cref{lem:open_kl_regret} establishes an open-system dynamic
$1/2$-approximation-regret bound for SubMAPL relative to the time-varying stagewise PME maximizers. These results characterize the stagewise tracking performance of SubMAPL along its induced state trajectory.
However, the objective in \cref{prob:main} is the optimal finite-horizon
policy value, which need not coincide with the cumulative stagewise optimal values. Thus, we quantify this discrepancy through the mismatch term $\Delta_T$ in \eqref{eq:horizon_gap_decomposition}. Under
\cref{ass:telescoping_coverage}, the policy-independent time-expanded
submodular representation connects the stagewise benchmark to the
finite-horizon optimum, yielding the approximation guarantee in
\cref{thm:telescoping_global}.


\section{Numerical Experiments}
\label{sec:experiments}

\begin{figure*}[t]
    \centering
    \includegraphics[width=\linewidth]{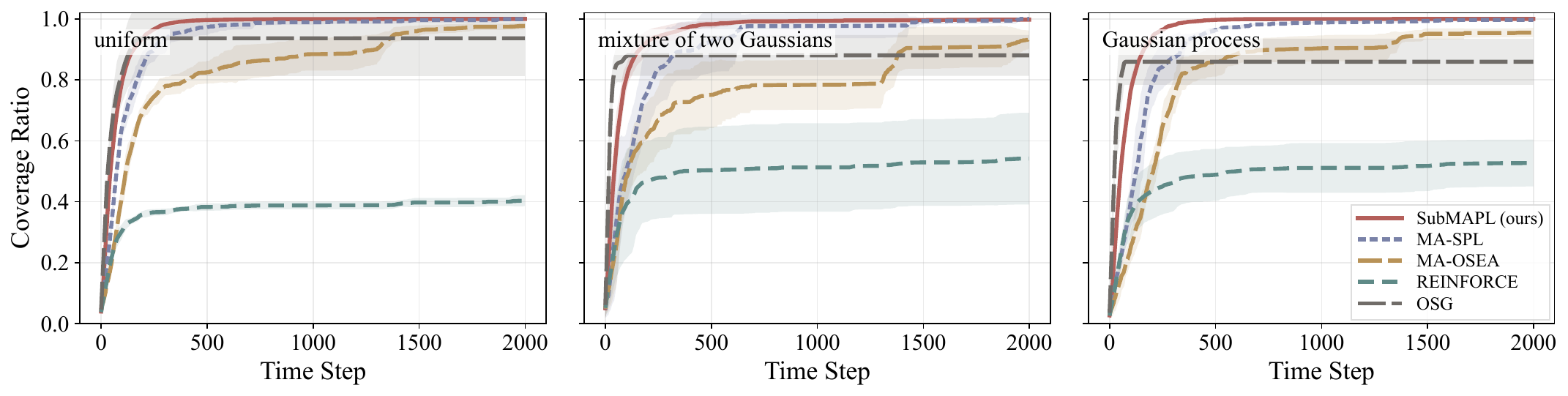}
    \captionsetup{font={small}}
    \caption{Coverage ratio under the controlled $5\to2\to5$ active-agent profile. From left to right: uniform, mixture of two Gaussians, and Gaussian process fields. Curves show the means, and shaded regions denote $95\%$ confidence intervals over shared evaluation scenarios.}
    \label{fig:main_comparison}
\end{figure*}

We evaluate SubMAPL on a multi-agent coverage task \cite{prajapat2024submodular}. SubMAPL is trained with a fixed (time-invariant) agent set and evaluated on open systems with a time-varying agent set.

\subsection{Experimental Setup}
\label{sec:exp_setup}

The workspace $\mathcal{V}$ is a $30\times30$ grid with $N_{\max}=5$ agents. To model spatially heterogeneous information, we weigh grid cells according to a probability distribution. We consider a uniform field, a mixture of two isotropic Gaussian components, and a positive log-Gaussian process field generated using a separable radial basis function (RBF) kernel. Cell weights remain constant throughout each experiment. All methods are evaluated with the same cell weights and initial agent positions.

The global state $s_t$ contains the active-agent set
$\mathcal N_t$, the current agent positions, and the accumulated coverage set $\mathcal C_{t-1}^{\rm hist} = \bigcup_{k=1}^{t-1}\mathcal C_k(A_k;s_k)$,
which records all cells sensed before the action at time $t$. The complete coverage history is maintained by the environment for state transitions and utility evaluation but is not directly available to individual agents during decentralized execution. Each agent observes only a local encoding of the coverage status described below.

Each agent receives a discrete local observation $o_{i,t}$. It contains: 1) the index of the agent's current $6\times6$ region in a partition of the workspace into $25$ regions; 2) the relative positions of at most two nearest active neighbors within the communication radius $r_{\rm com}=2$; 3) the states of the agent's current cell and its four axially adjacent cells, where each cell is classified as already covered or outside the workspace, uncovered with low information density, or uncovered with high information density; and 4) the agent's previous action. Thus, an agent observes only nearby coverage and neighbor information, rather than the complete coverage history or the positions of all agents.

At each step $t$, active agent $i\in\mathcal N_t$ selects a motion action $a_{i,t}\in\mathcal A_{i}=\{\textnormal{idle},\textnormal{left},\textnormal{right},\textnormal{up},\textnormal{down}\}$, which determines its next position. Let $D_{i}(a_{i,t};s_t)$ denote the sensing region of agent $i$ after applying action $a_{i,t}$ at state $s_t$, which includes all cells within Chebyshev distance $r_{\rm cov}=1$ from the agent's new position. Hence, $D_{i}(a_{i,t};s_t)$ forms a $3\times3$ neighborhood and determines the information collected by agent $i$.

The cells sensed as a result of executing $A_t$ at state $s_t$ are 
\begin{equation}
    \mathcal C_t(A_t;s_t) = \bigcup_{i\in\mathcal N_t}
D_{i}(a_{i,t};s_t). 
\end{equation}
Let
$\mathcal C_{t-1}^{\rm hist} = \bigcup_{k=1}^{t-1}\mathcal C_k(A_k;s_k)$ denote all cells covered before the action at time $t$. We define the stage utility as the weighted information collected from  uncovered cells:
\begin{equation}
    F_t(A_t;s_t)=\sum_{v\in\mathcal C_t(A_t;s_t)\setminus\mathcal C_{t-1}^{\rm hist}}\rho(v),
\end{equation}
where $\rho(v)\geq 0$ is the weight of cell $v\in\mathcal{V}$. Utility $F_t(\cdot;s_t)$ is a normalized, monotone, and submodular set function over the active agent-action ground set $\Omega_t$, satisfying \cref{ass:standing_conditions} \cite{krause2014SubmodularFM,sun2019exploiting}.
We measure performance as the weighted coverage ratio
\begin{equation}
    \mathrm{CovRatio}_t =  \frac{  \sum_{v\in\mathcal C_t^{\rm hist}}\rho(v) }{ \sum_{v\in\mathcal V}\rho(v)}.
\end{equation}

We compare our proposed algorithm SubMAPL (\cref{alg:submapl}) with four benchmarks.
\begin{description}[leftmargin=*]
    \item[\emph{MA-SPL}~\cite{zhang2025effective}:] a policy-based continuous-extension method that employs a submodular surrogate gradient, a minimum-marginal-gain correction, neighbor consensus, and Euclidean projection.
    \item[\emph{MA-OSEA}~\cite{zhang2025near}:] a multilinear-extension method based on curvature-aware surrogate gradients, neighbor consensus, uniform mixing, and a projection-free KL update.
    \item[\emph{REINFORCE}~\cite{williams1992simple}:] a tabular return-based policy-gradient method trained using the global newly covered information as the stage reward.
    \item[\emph{OSG}~\cite{xu2023online}:] a centralized online greedy method that sequentially selects one action for each agent according to their current marginal coverage gains following a fixed order.
\end{description}

SubMAPL and REINFORCE are trained in a closed-system setting for
$3000$ episodes with horizon $T=100$ steps. Each evaluation trajectory has length $T=2000$.
MA-SPL and MA-OSEA are initialized with uniform policies and perform their native online updates throughout the evaluation horizon. At each time step, OSG processes the active agents in a fixed order, and each active agent greedily maximizes $F_t$ over the active agent set at each step $t$.
A complete communication graph is used for MA-SPL and MA-OSEA. SubMAPL does not perform inter-agent
communication or consensus during decentralized execution. Each agent instead locally observes the relative positions of at most the two nearest active agents within communication radius of $r_{\rm com}=2$.

For the methods trained in the closed-system setting, each agent retains its finite-logit policy table while inactive and restores it upon reactivation. 
We use five training seeds and five evaluation scenarios. For SubMAPL and REINFORCE, results are first averaged over training seeds within each evaluation scenario.
The curves report averages across the five evaluation scenarios, and the shaded regions delimit $95\%$ confidence intervals. 

\subsection{Results}
\label{sec:exp_results}

\begin{table}[t]
    \centering
    \caption{Normalized areas under the coverage curves in the
    open-system evaluation. Higher values indicate faster overall
    information acquisition.}
    \label{tab:open_normalized_auc}
    \setlength{\tabcolsep}{3.5pt}
    \renewcommand{\arraystretch}{1.10}
    \begin{tabular}{@{}lcccc@{}}
        \toprule
        Method
        & Uniform
        & \shortstack{Mixture of two\\Gaussians}
        & \shortstack{Gaussian\\process}
        & Average \\
        \midrule
        SubMAPL
        & \textbf{0.966}
        & \textbf{0.960}
        & \textbf{0.963}
        & \textbf{0.963} \\
        MA-SPL
        & 0.939
        & 0.914
        & 0.918
        & 0.924 \\
        MA-OSEA
        & 0.843
        & 0.775
        & 0.831
        & 0.816 \\
        REINFORCE
        & 0.377
        & 0.499
        & 0.487
        & 0.454 \\
        OSG
        & 0.917
        & 0.873
        & 0.850
        & 0.880 \\
        \bottomrule
    \end{tabular}
\end{table}

We first evaluate performance under mild changes in the active agent set $\mathcal{N}_t$, and then conduct a stress test with strongly time-varying agent participation.

\subsubsection{Main Comparison}

\begin{figure}[t]
    \centering
    \includegraphics[width=0.8\columnwidth]{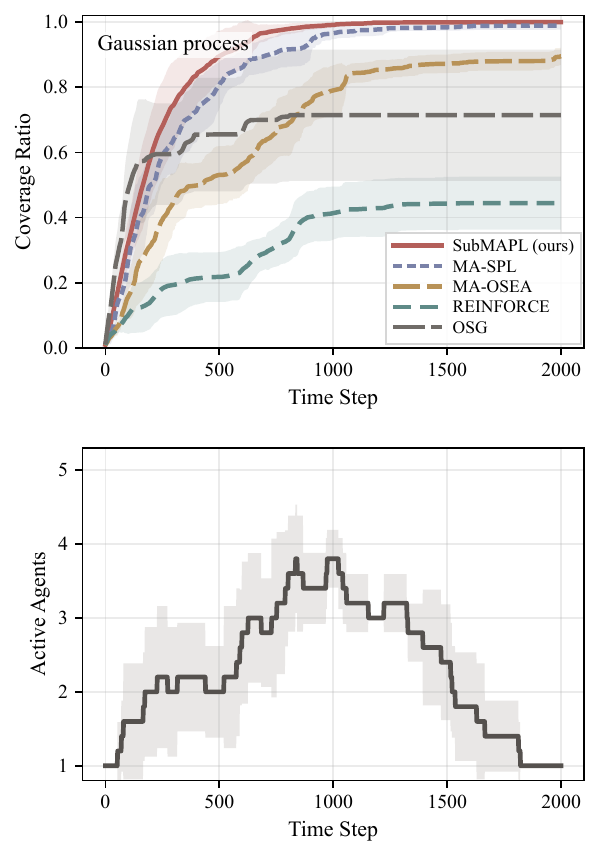}
    \captionsetup{font={small}}
    \caption{Performance under random agent participation on the
    Gaussian-process field. Top: coverage ratios of the five methods. Bottom: the active-agent profile shared by all methods, reported as the mean with $95\%$ confidence intervals over five shared evaluation scenarios.}
    \label{fig:open_random_gp}
\end{figure}

In the main comparison, all five agents are active for $t=0,\ldots,699$, three agents are temporarily unavailable for $t=700,\ldots,1299$, and all five are active again from $t=1300$. 
\cref{fig:main_comparison} compares the five methods under the same open-system profile. SubMAPL achieves the fastest overall information acquisition in all three landscapes. \cref{tab:open_normalized_auc} reports the normalized areas under the coverage curves for all five methods. SubMAPL achieves the largest area in every information-density landscape, with an average
normalized area of $0.963$, compared with $0.924$ for the strongest competing method, MA-SPL. This advantage indicates that SubMAPL acquires information more rapidly throughout the evaluation horizon.

The online optimization baselines exhibit different transient behavior. OSG initially acquires information rapidly by selecting actions from current marginal gains, but its sequential myopic decisions eventually plateau below complete coverage. MA-OSEA improves after the active population is restored at $t=1300$, especially in the nonuniform fields, but remains slower than MA-SPL and SubMAPL. REINFORCE is substantially slower and shows the clearest loss of progress during the interval with only two active agents. The results support the benefit of combining this categorical relaxation with full local-action marginal feedback and KL-mirror policy updates in the tested open-system coverage task.

\subsubsection{Open-system Robustness}

\begin{table}[t]
    \centering
    \caption{Performance under the random agent participation on the Gaussian-process information field.}
    \label{tab:strong_random_robustness}
    \setlength{\tabcolsep}{3.2pt}
    \renewcommand{\arraystretch}{1.10}
    \begin{tabular}{@{}lcccc@{}}
        \toprule
        Method
        & \shortstack{Normalized\\area}
        & \shortstack{Final\\coverage}
        & \shortstack{$95\%$ coverage\\successes}
        & \shortstack{Mean\\$T_{0.95}$} \\
        \midrule
        SubMAPL
        & \textbf{0.887}
        & \textbf{0.999}
        & \textbf{5/5}
        & \textbf{580} \\
        MA-SPL
        & 0.844
        & 0.989
        & 5/5
        & 904 \\
        MA-OSEA
        & 0.683
        & 0.894
        & 0/5
        & -- \\
        REINFORCE
        & 0.338
        & 0.444
        & 0/5
        & -- \\
        OSG
        & 0.665
        & 0.714
        & 0/5
        & -- \\
        \bottomrule
    \end{tabular}

    \vspace{2pt}
    \parbox{\columnwidth}{\footnotesize
    The normalized area is the area under the coverage-ratio curve
    normalized by the evaluation horizon. $T_{0.95}$ denotes the first time step at which the coverage ratio reaches $0.95$ and is averaged only over successful scenarios. The mean normalized-area advantage of SubMAPL over MA-SPL is $0.043$, with a $95\%$ confidence interval of $[0.009,0.078]$.}
\end{table}

We next evaluate the transferability of the closed-system trained policies under stronger stochastic variation in agent participation. The experiment uses the Gaussian process field.
One agent remains active throughout the evaluation horizon, whereas each of the other four agents is active over a randomly generated contiguous time interval. All methods are evaluated under the same active-agent schedule within each evaluation scenario.

\cref{fig:open_random_gp} shows that SubMAPL achieves the fastest cumulative information acquisition despite substantial changes in both the size and composition of the active-agent set. 
\cref{tab:strong_random_robustness} shows that SubMAPL achieves the largest normalized area and the highest final coverage.  Both SubMAPL and MA-SPL reach $0.95$ coverage in all five scenarios, but SubMAPL reaches this level considerably earlier, with a mean hitting time of $580$ steps compared with $904$ steps for MA-SPL. 

SubMAPL reuses observation-conditioned local policy tables learned before deployment and restores the corresponding policy whenever an agent becomes active. Therefore, it avoids restarting the decision process after each change in participation. MA-SPL continues to adapt online, but its slower transient response reduces the information accumulated early in the finite horizon. These results indicate that the learned SubMAPL policies remain effective under substantial random changes in the size and composition of the active-agent set, without deployment-time policy updates.


\section{Conclusion}
\label{sec:conclusion}

This paper studied decentralized policy learning for open multi-agent task allocation with monotone submodular team utilities. We introduced the partition multilinear extension (PME) as a policy-based continuous representation of the expected stage utility induced by factorized categorical policies and developed SubMAPL, a KL-mirror policy-learning method using the stochastic local gradient as feedback over each agent's feasible action set. We showed that this feedback provides unbiased PME coordinate-gradient information and that the KL-mirror update is exactly equivalent to an additive logit update under tabular softmax parameterization. We established open-system approximation-regret guarantees and finite-horizon performance bounds. Numerical experiments on information coverage showed that SubMAPL policies remain effective under random changes in agent participation and outperform online submodular-coordination and policy-gradient baselines. Future work will consider richer communication constraints, neural policy approximation, and long-horizon value-aware objectives.


\bibliographystyle{IEEEtran}
\bibliography{refs}

\end{document}